\documentclass{article}
\usepackage{fullpage}

\usepackage[utf8]{inputenc}

\usepackage{amsthm, amsmath, amssymb, amsfonts}
\usepackage{thm-restate}

\usepackage{graphicx}
\graphicspath{{figs/}}
\usepackage{caption}
\usepackage{subcaption}
\usepackage[capitalise]{cleveref}
\usepackage{xspace}
\usepackage{verbatim}

\usepackage{enumitem}   
\setlist{topsep=3pt, itemsep=3pt}

\usepackage{float}
\usepackage[normalem]{ulem}

\newtheorem{theorem}{Theorem}[section]
\newtheorem{lemma}[theorem]{Lemma}

\newtheorem{coro}[theorem]{Corollary}

\newtheorem{claim}[theorem]{Claim}

\newcommand{\bbR}{{\mathbb{R}}}
\newcommand{\bbZ}{{\mathbb{Z}}}

\usepackage{algorithm}
\usepackage{algorithmicx}
\usepackage{algpseudocode}
\algtext*{EndWhile}
\algtext*{EndIf}
\algtext*{EndFor}
\algtext*{EndFunction}
\algtext*{EndProcedure}

\newcommand{\FDP}{\textrm{FDP}\xspace}

\newcommand{\mst}{{\mathsf{mst}}}
\newcommand{\sfleft}{{\mathsf{left}}}
\newcommand{\sfright}{{\mathsf{right}}}
\newcommand{\etal}{{\it et~al.}}

\newcommand{\TSP}{{\mathsf{TSP}}}
\newcommand{\CVRP}{{\mathsf{CVRP}}}

\newcommand{\calA}{{\mathcal{A}}}
\newcommand{\calP}{{\mathcal{P}}}
\newcommand{\calQ}{{\mathcal{Q}}}
\newcommand{\calS}{{\mathcal{S}}}
\newcommand{\calI}{{\mathcal{I}}}

\newcommand{\sfL}{{\mathsf{L}}}
\newcommand{\sfR}{{\mathsf{R}}}

\allowdisplaybreaks

\newcommand{\cost}{{\mathsf{cost}}}

\newcommand{\floor}[1]{\left\lfloor#1\right\rfloor}
\newcommand{\ceil}[1]{\left\lceil#1\right\rceil}
\mathchardef\mhyphen="2D

\title{Online Food Delivery to Minimize Maximum Flow Time}
\author{
Xiangyu Guo\qquad Shi Li\\ 
University at Buffalo \\
\texttt{\{xiangyug,shil\}@buffalo.edu}
        \and 
        Kelin Luo\\
        Eindhoven University of Technology\\
        \texttt{k.luo@tue.nl}
        \and 
        Yuhao Zhang\\ 
        Shanghai Jiao Tong University\\
        \texttt{zhang\_yuhao@sjtu.edu.cn}
}
\date{}

\begin{document}
    \maketitle
    
    \thispagestyle{empty}
    
    \begin{abstract}


We study a common delivery problem encountered in nowadays online food-ordering platforms: Customers order dishes online, and the restaurant delivers the food after receiving the order. Specifically, we study a problem where $k$ vehicles of capacity $c$ are serving a set of requests ordering food from one restaurant. After a request arrives, it can be served by a vehicle moving from the restaurant to its delivery location. 
We are interested in serving all requests while minimizing the maximum flow-time, i.e., the maximum time length a customer waits to receive his/her food after submitting the order.  The problem also has a close connection with the broadcast scheduling problem with maximum flow time objective.

We show that the problem is hard in both offline and online settings even when $k = 1$ and $c = \infty$: There is a hardness of approximation of $\Omega(n)$ for the offline problem, and a lower bound of $\Omega(n)$ on the competitive ratio of any online algorithm, 
where $n$ is number of points in the metric. 

We circumvent the strong negative results in two directions. Our main result is an $O(1)$-competitive online algorithm for the uncapaciated (i.e, $c = \infty$) food delivery problem on tree metrics; we also have negative result showing that the condition $c = \infty$ is needed.
Then we consider the speed-augmentation model, in which our online algorithm is allowed to use $\alpha$-speed vehicles, where $\alpha \geq 1$ is called the speeding factor. We develop an exponential time $(1+\epsilon)$-speeding $O(1/\epsilon)$-competitive algorithm for any $\epsilon > 0$. A polynomial time algorithm can be obtained with a speeding factor of $\alpha_\TSP + \epsilon$ or $\alpha_\CVRP + \epsilon$, depending on whether the problem is uncapacitated. Here $\alpha_\TSP$ and $\alpha_\CVRP$ are the best approximation factors for the traveling salesman (TSP) and capacitated vehicle routing (CVRP) problems respectively. We complement the results with some negative ones.

\end{abstract}
    
    \newpage
    \setcounter{page}{1}
    
    \section{Introduction}  
Online food-ordering-and-delivery services (e.g., UberEats and Doordash) have become more and more popular in the past decade. Customers can order dishes online and wait for the restaurant to deliver the food to their home. Arguably one of the most important factors for service quality is the \emph{flow time}, i.e., how long a customer waits to receive the food after submitting his/her order. We formalize the problem as the following \emph{Online Food Delivery Problem} (Online \FDP). Let $(V, d)$ be a metric space with $|V| = n$, where $d$ is a metric on $V$. Let $o\in V$ be the \emph{depot} that has $k$ unit-speed vehicles, each with a capacity $c \in \bbZ_{ > 0} \cup \{\infty\}$. The requests arrive online, where a request $\rho=(r_\rho, v_\rho)$ is released at time $r_\rho$ with delivery location $v_\rho$.   The goal of the online problem is to schedule the vehicles to deliver the food to their delivery locations (i.e., to serve the requests), satisfying the requirement that a vehicle can only carry food for at most $c$ requests at any time. The objective of the problem is to minimize the maximum flow time. This model captures the real-world scenario that one restaurant owns several vehicles and need to deliver food to customers, and the restaurant want to make all customers satisfied, i.e, wait for a short period of time.


 The food delivery problem (FDP) is closely related to many vehicle routing problems. For example, the uncapacitated (i.e., $c=\infty$) single-vehicle (i.e., $k = 1$) FDP with all requests released at time 0 already captures a variant of the traveling salesman path problem. When the vehicle has finite capacity $c$, the FDP just mentioned is also known as the Capacitated Vehicle Routing Problem (CVRP). A more general setting is studied under the name Dial-a-Ride Problem (DaRP), where besides the delivery location, each customer can specify its own pickup location. To satisfy the request, a vehicle has to take the customer from the pickup location to its delivery location. Most results on offline CVRP/DaRP focused on the total travel distance objective  \cite{altinkemer1987heuristics,haimovich1985bounds,charikar1998finite,gupta2010dial}; while in the online model, much effort has been devoted to completion time objectives like average completion time \cite{hwang2018online,bienkowski2021traveling} or makespan \cite{FEUERSTEIN200191,ascheuer2000online}. Compared with previous studies, we are focusing on a much harder and less-studied objective --- maximum flow time, in the online setting. The only theoretical result we know regarding maximum flow time is a lower bound for online DaRP, where it is shown the single-vehicle finite-capacity DaRP does not admit $O(1)$-competitive algorithms \cite{krumke2005minimizing}. 

Another motivation for FDP is from the seemingly unrelated broadcast scheduling problem. In this problem, a server holds $n$ pages of varying sizes and requests are released over time, each of which is a query on one of the pages. The server can \emph{broadcast} a page to all requests on that same page, which takes time equal to the page size. The objective is to minimize the maximum flow time. It is easy to reduce the broadcast scheduling problem to uncapacitated single-vehicle FDP on a star, where a page of size $s$ in the broadcast scheduling problem corresponds to an edge of length $s/2$ in the FDP problem.\footnote{In FDP, we could define the completion time of a request as the time the vehicle returns to the depot after serving the request. The maximum flow time objective for the two versions differs by an $O(1)$ factor.}
It is known that FIFO gives $O(1)$-competitiveness for the broadcast scheduling problem \cite{chekuri2009minimizing,chang2011broadcast}, and thus the single-vehicle uncapacitated FDP on stars.   One of our main results in the paper is an $O(1)$-competitive algorithm for uncapacitated FDP on \emph{trees}, generalizing the results to tree metrics and the multiple-vehicle case.   Mapping to the broadcast scheduling application, this gives the following more general setting.  There is a tree rooted at $o$, where the pages correspond to the leaves of a tree, and each internal node corresponds to a setup procedure that takes some time to run. To broadcast a set of pages, in addition to the broadcasting time for each page, we have to spend time on running the setup procedures correspondent to the ancestor nodes of the pages.  

\subsection{Our Results} We state our results in this section. The formal definition of the offline and online food delivery problem (FDP) can be found in Section \ref{sec:prelim}.  In all the theorems below, $n$ is the number of points in the metric space. \medskip

We first show that our problem also suffers from the maximum flow time objective, both in offline and online settings, even if we focus on the single-vehicle uncapacitated case. 
\begin{theorem}
    \label{thm:LB}
    The following statements hold for the single-vehicle uncapacitated FDP:
    \begin{enumerate}[label=(\ref{thm:LB}\alph*),leftmargin=*, itemsep=0pt]
        \item 
        It is NP-hard to approximate the offline problem within a factor of $o(n)$.
        \item 
        There is no $o(n)$-competitive algorithm for the online problem, even if it can run in exponential time and the metric is on a bounded-pathwidth planar graph.
    \end{enumerate}
\end{theorem}
To circumvent the strong lower bounds, we first study if the problem admits $O(1)$-competitive algorithms on special metrics. 
Given that (\ref{thm:LB}b) holds even for bounded-pathwidth planar graphs, a natural candidate family of metrics is the tree metrics. Our main result of the paper is that uncapacitated FDP on trees admit $O(1)$-competitive online algorithms:

\begin{theorem}[Main Result]\label{thm:ALG-on-trees}
There's an efficient $O(1)$-competitive online algorithm for uncapacitated $\FDP$ on trees.
\end{theorem}

The results in Theorem~\ref{thm:LB} and \ref{thm:ALG-on-trees} are for uncapacitated version of the problem ($c = \infty$).  We now turn to the general capacitated case.  One can ask if the algorithm in Theorem~\ref{thm:ALG-on-trees} works for general $c$. Unfortunately, we show that this is impossible even for  a very special case. Below $R$ is the set of requests in the instance (See Section~\ref{sec:prelim}):
\begin{theorem}
    \label{thm:LB-capacity}
    There is no $o(\sqrt{|R|})$-competitive online algorithm for single-vehicle FDP with $c = 2$ and on a tree with 5 vertices.
\end{theorem}

Despite all the negative results, we show that the problem admits good competitive algorithms under the \emph{speed augmentation} model. 
In this model we allow vehicles in our algorithm to run faster than normal vehicles, while we compare the maximum flow time achieved by our algorithm using faster vehicles against the optimum maximum flow time that can be achieved with normal vehicles.  The speeding factor defines how faster our vehicles, and our goal is to decide the tradeoff between the speeding factor and the competitive (approximation) ratio of our algorithm.
There are two reasons we consider this model. First, in the optimum solution of the hard instance for (\ref{thm:LB}b), the vehicle is always busy. One mistake made by the online algorithm can delay the whole schedule. This situation rarely happens in practice, where one should expect the vehicles to have a reasonable amount of idle time.  Second, the travel time between two points in practice is often an estimate and affected by many factors such as the driver's skill and the traffic condition. So, a solution that becomes very bad if the speeds of vehicles decrease slightly should be considered as very fragile. 
We remark that the speed augmentation model is also popular in many scheduling problems (see Section~\ref{subsec:more-related}).


Let $\alpha_\TSP$ be the infinum of all values $\alpha$ such that there is a (polynomial time) $\alpha$-approximation algorithm for TSP.  With the breakthrough result of Karlin, Klein, and Oveis Gharan~\cite{karlin2021slightly}, we know that $\alpha_\TSP$ is strictly smaller than $1.5$. Define $\alpha_\CVRP$ similarly for the CVRP problem (See Section~\ref{sec:prelim} for formal definition). It is known that $\alpha_\CVRP \leq \alpha_\TSP + 1$.  We give our results for the food delivery problem with speeding below; the formal definition of the speed augmentation setting can be found in Section~\ref{sec:prelim}.  

\begin{restatable}{theorem}{AlgSpeeding}\label{thm:ALG-speeding}
For any small enough constant $\epsilon > 0$, there are
\begin{enumerate}[label=(\ref{thm:ALG-speeding}\alph*),leftmargin=*,itemsep=0pt]
    \item an exponential time $(1+\epsilon)$-speeding $O(1/\epsilon)$-competitive online algorithm for $\FDP$, 
    \item a polynomial time $(\alpha_\TSP + \epsilon)$-speeding $O(1/\epsilon)$-competitive online algorithm for uncapacitated FDP, and
    \item a polynomial time $(\alpha_\CVRP + \epsilon)$-speeding $O(1/\epsilon)$-competitive online algorithm for (capacitated) FDP.
\end{enumerate}
\end{restatable}


\begin{theorem}
    \label{thm:LB-speeding}
    The following statements hold.
    \begin{enumerate}[label=(\ref{thm:LB-speeding}\alph*),leftmargin=*,itemsep=0pt]
    \item There is no $(1+\epsilon)$-speeding $o(1/\epsilon)$-competitive online algorithm for FDP. 
    \item For any constant $\alpha \in [1, \alpha_\TSP)$, there is no (polynomial time) $\alpha$-speeding $o(\sqrt{n})$-approximation algorithm for uncapaciated FDP.
    \item For any constant $\alpha \in [1, \alpha_{\CVRP})$, there is no (polynomial time) $\alpha$-speeding $o(\sqrt{n})$-approximation algorithm for (capaciated) FDP.
    \end{enumerate}
\end{theorem}
Therefore, by (\ref{thm:LB-speeding}a), the $O(1/\epsilon)$-dependence on $\epsilon$ in (\ref{thm:ALG-speeding}a) is needed. (\ref{thm:LB-speeding}b) and (\ref{thm:LB-speeding}c) state that the speeding factors in (\ref{thm:ALG-speeding}b) and (\ref{thm:ALG-speeding}c) are almost tight. 

\subsection{An Overview of Techniques}
    The hardness results in (\ref{thm:LB}a), (\ref{thm:LB-speeding}b) and (\ref{thm:LB-speeding}c) are proved using simple reductions from TSP or CVRP, in which we repeat a TSP/CVRP instance multiple times. Since an efficient algorithm can not find the best TSP/CVRP tour, it needs to use a longer tour for each instance, which creates a delay in the schedule. The delays for all instances will accumulate, resulting in a bad flow time. The hardness remains if the speeding factor is not big enough. 
        
    The lower bounds in (\ref{thm:LB}b), (\ref{thm:LB-speeding}a) and Theorem~\ref{thm:LB-capacity} are proved using a same idea. We build a base instance satisfying the following property. At the beginning of the time horizon, the online algorithm needs to make a decision between two choices. If it made the incorrect choice, the total length of trips it uses to satisfy all requests will be 1 unit time longer than the case if it made the correct choice. The catch is, the algorithm will only know which choice is the correct one until near the end of the time horizon.  So, it has to either wait for almost the whole time horizon, which will incur a large maximum flow time, or it will spend 1 more unit time on the trips, delaying the schedule. By repeating the base instance many times, the delay of the online algorithm will accumulate, resulting in a large maximum flow time.

    The overview of the algorithms for online uncapacitated FDP on trees will be given at the beginning of Section~\ref{sec:tree-single}. The online algorithms in Theorem~\ref{thm:ALG-speeding} are based on a simple idea. We wait for $\gamma F$ time units ($F$ is a given upper bound of the optimal flow time. See section~\ref{sec:prelim} for details.) for some $\gamma = \Theta\left(\frac 1\epsilon\right)$. Then we know that the optimum solution can serve all the requests arrived in the $\gamma F$-length interval using $(\gamma + O(1))F$ time.  With $(1+\epsilon)$-speeding, the online algorithm can serve them in $\gamma F$ time. Thus, we can always guarantee an $O(F/\epsilon)$-flow time for all requests.  If the algorithm has to be efficient, we need to lose a speeding factor of $\alpha_\TSP + \epsilon$ or $\alpha_{\CVRP} + \epsilon$, depending on whether the instance is uncapacitated.



\subsection{Other Related Works}
\label{subsec:more-related}

Most of previous results on CVRP and DaRP focus the single vehicle case and the total travel distance objective, which is equivalent to the makespan. Unless specified otherwise, all the results surveyed below for CVRP and DaRP are for this case. \medskip

\noindent{\bfseries Capacitated Vehicle Routing Problem (CVRP)}\ \  
The CVRP is mostly studided in the offline setting. It admits an $(\alpha_\TSP + 1)$-approximation, where $\alpha_\TSP$ is the approximation ratio for traveling salesman problem (TSP) \cite{altinkemer1987heuristics,haimovich1985bounds,karlin2021slightly}. A more general version of CVRP is studied in the literature: Each request has a demand, and we require the total demand of all requests satisfied in a single trip made by the vehicle is at most $c$. In the same papers, Altinkemer and Gavish \cite{altinkemer1987heuristics} and Haimovich and Rinnooy Kan \cite{haimovich1985bounds} gave an approximation algorithm with ratio $\alpha_\TSP+2$ for the general CVRP, which stood for over 30 years. Very recently, the approximation ratio for the problem has been improved to $\alpha_\TSP+2(1-\epsilon)$ by Blauth~\etal\ \cite{blauth2020improving}, for some small constant $\epsilon > 0$. There are also improved ratios on special metrics like Euclidean plane \cite{das2010quasi,khachay2016ptas} and tree metrics \cite{labbe1991capacitated,wu2020capacitated}.

One variant of CVRP that is related to the flow time objective is the \emph{CVRP with bounded delay}: There is a \emph{delay constraint} that any request $\rho$ spends at most $\beta d(r,v_\rho)$ time \emph{on the vehicle}, where $v_\rho$ denotes the drop-off location of $\rho$. The goal is to find a minimum length route that satisfies all the delay constraints. For this variant, G\o rtz~\etal\ \cite{gortz2009minimum} gave a $\left(2.5+\frac{3}{\beta-1}\right)$-approximation algorithm for single-vehicle CVRP with bounded delay. \medskip

\noindent{\bfseries Dial-a-Ride Problem (DaRP)}\ \  Like CVRP, the most studied setting for offline DaRP is also for the single-vehicle case with makespan objective, for which the best known algorithm achieves $\tilde{O}(\min\{\sqrt{n},\sqrt{c}\})$ approximation ratio~\cite{charikar1998finite, gupta2010dial}, where $n$ is the number of requests. The best known lower bound is only APX-hardness.  The \emph{online} DaRP is mostly studied with objectives like makespan\cite{ascheuer2000online} or total/weighted completion time\cite{FEUERSTEIN200191,krumke2003news,bienkowski2021traveling}: all of these objectives admit small constant competitive-ratio algorithm. However, if we turn to minimizing the maximum flow-time, there exists no $o(n)$-competitive algorithm even on a 3-point uniform metric with a unit-capacity vechile~\cite{krumke2005minimizing}. 

A variant of DaRP studied in the literature that seems much easier is the preemptive DaRP, in which the vehicle is allowed to temporarily unload the cargo it carries in the middle of a trip, and pick it up later to resume delivery. The best approximation ratio for the problem with single vehicle is $O(\log n)$ \cite{charikar1998finite}, and there is a hardness of $\Omega(\log^{1/4}n)$ hardness of approximation~\cite{gortz2006hardness}. For the multi-vehicle case, G\o rtz~\etal\cite{gortz2009minimum} gave an $O(\log^3 n)$ for the preemptive DaRP, and an $O(\log n)$-approximation when vehicles have infinite capacity.  
\medskip

\noindent{\bfseries Flow-Time Scheduling}\ \  
The maximum flow time objective is studied in many scheduling problems. The simple FIFO strategy is known to achieve the best competitive ratio 2~\cite{mastrolilli2004scheduling} in the identical machine setting, and constant competitive algorithm exists even in related machine setting \cite{bansal2016minimizing} (i.e., machines have different speed). However, the approximability changes dramatically if we switch the objective to the (weighted) sum of flow-time: if no preemption is allowed, the problem is $\Omega(m)$-hard in the online setting where $m$ is the number of machines, and $\Omega(m^{1/2-\epsilon})$-hard in the offline setting~\cite{KellererTW96}. Even when preemption is allowed, offline $O(1)$-approximations (for the single-machine setting) are known only very recently~\cite{batra2018constant, feige2019polynomial, rohwedder2020}, and there's an $\omega(1)$ lower bound for the online case~\cite{bansal2009weighted}.

For online broadcast scheduling, the maximum flow time objective also appears to be easier than other flow time objectives: max flow time admits 2-competitive algorithm\cite{bartal2000minimizing,chang2011broadcast,chekuri2009minimizing}, while average flow time has very strong lower bounds  $\Omega(n)$ for deterministic algorithms\cite{kalyanasundaram2000scheduling} and $\Omega(\sqrt{n})$ for randomized algorithms \cite{bansal2005approximating}.

\subsection{Organization} The remaining part of the paper is organized as follows. In Section~\ref{sec:prelim}, we formally define the offline and online food delivery problem, and the speed augmentation model. We prove our main result, Theorem~\ref{thm:ALG-on-trees}, by giving the $O(1)$-competitive algorithm for uncapacitated FDP on trees in Sections~\ref{sec:tree-single} and~\ref{sec:tree-multi}. The two sections consider the single-vehicle and multiple-vehicle cases respectively. Then in Section~\ref{sec:speeding}, we prove Theorem~\ref{thm:ALG-speeding} by giving the online algorithms for FDP with speeding.   The proofs of Theorem~\ref{thm:LB} and \ref{thm:LB-speeding} are mixed in Appendices~\ref{sec:hardness} and \ref{sec:LB}, which focus on hardness of offline algorithms and lower bounds for online algorithms respectively. The proof of Theorem~\ref{thm:LB-capacity} is given in Appendix~\ref{sec:LB-capacity}.

Finally, we remark that in all the proofs of lower bounds for competitive ratios, we assume the online algorithms are deterministic. However, it is not hard to extend the proofs to randomized algorithms as well.

    \section{Preliminary}
\label{sec:prelim}
We now define the food delivery problem formally, starting from the offline setting.  We are given a graph $G = (V, E)$ with edge lengths $\ell : E \to \bbR_{> 0}$, where $\ell(e)$ denotes the time needed to traverse the edge $e$ in either direction.\footnote{Equivalently we could use a metric $(V, d)$ to describe the travel times, but for many of our results it is more convenient to use the graph $G$ with edge lengths.} There is a special vertex $o \in V$ called the depot, which represents the restaurant.  We are given a number $k \geq 1$ of vehicles which are initially located at $o$, and capacity $c \in \bbZ_{>0} \cup \{\infty\}$ on the vehicles. When $k = 1$, we say the problem is single-vehicle, and when $c = \infty$, we say the problem is uncapacitated. There is a set $R$ of requests. Each request $\rho \in R$ is denoted by $\rho = (r_\rho, v_\rho)$, where $r_\rho \in \bbR_{\geq 0}$ and $v_\rho \in V$ are the arrival time and the delivery location of the request respectively.   

To describe the output of the problem, we need to define \emph{trips}.  A trip is defined by a triple $(t, (u_0 = o, u_1, u_2, \cdots, u_z = o), R')$, where $t \geq 0$ is the starting time of the trip, $(u_0 = o, u_1, u_2, u_3, \cdots, u_{z-1}, u_z = o)$ is  a (possibly complex) cycle in $G$ that starts and ends at $o$, and $R' \subseteq R$, $|R'| \leq c$, is the set of requests served by the trip.  So $z \geq 2$ and $(u_{z'-1}, u_{z'}) \in E$ for every $z' \in [z]$.  We require the following properties to hold for a trip.  First, $u_{z'} \neq o$ for every $z' \in [z-1]$; one can see easily soon that this is without loss of generality.  Then, for every served request $\rho \in R'$, we have $r_\rho \leq t$ and $v_\rho$ appears in the cycle.  If $\tilde z$ is the smallest index such that $u_{\tilde z} = v_\rho$, then we say $\rho$ is  served by the trip at time $t + \sum_{z' = 1}^{\tilde z} \ell(u_{z'-1}, u_{z'})$.   The completion time of the trip is defined as $t + \sum_{z' = 1}^{z} \ell(u_{z'-1}, u_{z'})$.  Abusing the definition slightly, sometimes we also use the word ``trip'' to denote the cycle $(u_0 = o, u_1, u_2, \cdots, u_z = o)$, with $t$ and $R'$ specified separately. 

The output of the food delivery problem contains $k$ sequences of trips correspondent to the itineraries of the
 $k$ vehicles. The following properties need to be satisfied. For every pair of adjacent trips in any of the $k$ sequences,  the starting time of the latter trip is at least the completion time of the former one.  Moreover, each request in $R$ is served by exactly one trip in the $k$ sequences; in other words, the sets of served requests in all trips of the $k$ sequences form a partition of $R$.  Let $t_\rho$ be the time that a request $\rho \in R$ is served (by the unique trip that serves it). We define the flow time of $\rho$ to be $t_\rho - r_\rho$.  The goal of the problem is to minimize the maximum flow time, i.e, $\max_{\rho \in R}(t_\rho - r_\rho)$.

So far we have defined the food delivery problem in the offline setting. In the online setting, $G, \ell, k$ and $c$ are given upfront, but the requests arrive online. Once we decided to start a trip at time $t$, then we have to complete the trip as planed.  Formally, we require that for any time $t \geq 0$, the prefixes of the $k$ sequences of trips with starting time at or before $t$ can only depend on the requests that arrive at or before $t$. \medskip

\noindent{\bf Speed Augmentation}\ \ In the speed augmentation model, we are allowed to use vehicles that run $\alpha$-times faster than the normal vehicles, where $\alpha \geq 1$ is called the speeding factor. Formally, we are solving the (offline or online) food delivery instance where the length function $\ell$ is changed to $\ell/\alpha$, while the competitive/approximation ratio of the algorithm is defined by comparing with the optimum solution without speeding, that is, w.r.t.\ the original length function $\ell$. So, we say an algorithm achieves $\alpha$-speeding $\beta$-competitive ratio (approximation ratio) if it obtains a solution with maximum flow time at most $\beta F$ for the instance where $\ell$ is replaced by $\ell/\alpha$, where $F$ is the optimum flow time for original instance.\medskip

\noindent{\bf Guessing the Optimum Maximum Flow Time Online}\ \  When designing online algorithms,  we shall use the standard doubling trick to assume that we are given an upper bound $F$ on the optimum maximum flow time of the instance, and the competitive ratio is defined by comparing to $F$.  That is, a $\beta$-competitive online algorithm has maximum flow time at most $\beta F$. If we have a $\beta$-competitive online algorithm $\calA$ under this setting, then we can obtain an $8\beta$-competitive algorithm $\calA'$ under the setting where $F$ is not given to us, while keeping the speeding factor and running time unchanged. In the paper, we are not trying to optimize the constant, and the factor of $8$ will be hidden in the $O(\cdot)$ notation. We show how this can be done in Appendix~\ref{appendix:guessing}.\medskip

\noindent{\bf Capacitated Vehicle Routing Problem (CVRP)}\ \ It is convenient for us to use the following definition of CVRP. We are given a graph $G = (V, E)$ with edge lengths $\ell: E \to \bbR_{>0}$, a depot $o \in V$, and an integer $c \in \bbZ_{ \geq 1} \cup \{\infty\}$. We are given a set $R$ of requests, each defined by its pickup location in $V \setminus \{o\}$. We need to find a (possibly-complex) cycle $(o, u_1, u_2, \cdots, u_z, o)$ with the minimum cost and associate each $u_{z'}, z' \in [z]$ with a subset $Q_{z'} \subseteq R$ of requests with pickup location $u_{z'}$. We require $\biguplus_{z' = 1}^{z}Q_{z'} = R$ and for every $1 \leq j \leq j' \leq z$ such that $u_{z'} \neq o$ for every $z' \in [j, j']$, we have $|\biguplus_{z' \in [j, j']}Q_{z'}| \leq c$.
%
If $c = \infty$, then clearly the problem becomes the well-known traveling salesman (TSP) problem.  Throughout the paper, we let $\alpha_\TSP$ ($\alpha_\CVRP$ resp.) be the infinum of all values $\alpha \geq 1$ such that there exists a polynomial time $\alpha$-approximation algorithm for TSP (CVRP resp.).

    	\section{Online Uncapacitated Single-Vehicle FDP on Tree Metrics}
	\label{sec:tree-single}
	In this section and the next one, we give the $O(1)$-competitive algorithm for online uncapacitated FDP on tree metrics, proving Theorem~\ref{thm:ALG-on-trees}. In this section, we consider the case $k = 1$, i.e, there is only one vehicle.  We separate this case from the general problem as its algorithm is simpler and the competitive ratio we obtain is better.
	
	First we setup some notations here.  Let $T = (V, E)$ be the tree and we assume it is rooted at the depot $o \in V$.  For every $v \in V$, we use $V_v$ to denote the set of descendants of $v$ (including $v$ itself). For an edge $e = (u, v)$ with $v$ being the child, we define $V_e = V_v$. We use $d$ to denote the metric induced by the tree $T$ with lengths $\ell(\cdot)$.  Given a set $X \subseteq V$ of vertices, we define $\mst(X)$ to be the cost of the minimal sub-tree of $T$ containing $X$ and $o$ (Notice that we require the tree to contain $o$).   Suppose we are further given an element $e \in V \cup E$.  We define $\mst_e(X) = \mst(V_e \cap X)$, which is the cost of the minimal sub-tree of $T$ containing $o$ and all vertices in $V_e \cap X$. So if $V_e \cap X = \emptyset$, then $\mst_e(X) = 0$; otherwise, the tree contains all the edges from $o$ to $e$ (including $e$ itself if it is an edge). By definition we have $\mst_o(X) = \mst(X)$. 
	
	Since most of the time, we deal with requests, it is convenient for us to use $\mst(R')$ and $\mst_e(R')$ to denote $\mst(\{v_\rho: \rho \in R'\})$ and $\mst_e(\{v_\rho: \rho \in R'\})$ respectively for a set $R' \subseteq R$ of requests.  Abusing notations slightly, sometimes we also use $\mst(X)$ to denote the actual sub-tree of $T$ achieving the cost $\mst(X)$. This also extends to $\mst_e(X), \mst(R')$ and $\mst_e(R')$.
	
	We now overview the algorithm for the online uncapacitated FDP problem on tree metrics, for both the single and multiple-machine cases.  We break the time horizon into intervals of length $F$: $\{[0, F), [F, 2F), \break [2F, 3F), [3F, 4F), ...  \}$ and index them by $1, 2, 3, 4, \cdots$. Recall that $F$ is an upper bound on the optimum maximum flow time, and the competitive ratio of our algorithm is defined against $F$. Let $R_i$ be the set of requests that arrive in interval $i$.  For each even integer $i$, we break $R_i$ into $R^\sfleft_i$ and $R^\sfright_i$, merge $R^\sfleft_i$ with $R_{i-1}$, merge $R^\sfright_i$ with $R_{i+1}$, so as to minimize some carefully designed objective.  So for each odd integer $i$, we have constructed a merged set $R'_i := R^\sfright_{i-1} \cup R_i \cup R^\sfleft_{i+1}$, which we call a \emph{bundle}.   The bundles then can be constructed in an online manner: the bundle $R'_i$ is determined by requests that arrive before time $(i+2)F$. 
	
	Then our online algorithm handles the bundles separately. When $k = 1$, we view each bundle $R'_i$ as a job of size $2\mst(R'_i)$; when $k \geq 2$, we break $R'_i$ into many \emph{groups} and treat each group $Q$ as a job of size $2\mst(Q)$. We think of all these jobs are released at time $(i+2)F$. We then assign the jobs to the vehicles online in a  greedy manner.  A crucial lemma we show is that the jobs have a small backlog: The jobs released in any interval $[aF, bF]$ have total size at most $(b-a)F + O(1)\cdot kF$, and each job has size $O(F)$. The properties guarantee that all the jobs complete with an $O(F)$ flow time. 
	
	\subsection{Description of Algorithm for Single-Vehicle Case}
	Now we formally state the algorithm for the case $k = 1$. 	For every integer $i \geq 1$,  let $R_i=\{\rho \in R: (i-1)F\le r_\rho < iF\}$ as stated.  Indeed, once we know a request is in $R_i$, we do not care about its precise arrival time anymore.  In the first step of the online algorithm, we partition the requests into  bundles $(R'_i)_{i \geq 1: i\text{ is odd}}$, using Algorithm~\ref{alg:partition}, where the bundle $R'_i$ is generated at time $(i+2)F$. For now let us focus on the case $k = 1$ in the algorithm. The formal definition for $\cost(\cdot|\cdot)$ is postponed to section~\ref{sec:tree-multi} since it's not used here. Throughout the section and the next one, we assume all undefined sets are set to $\emptyset$. See Figure~\ref{fig:sets}(a) for an illustration of the relationships between $R_i$, $R'_i$, $R^\sfleft_i$ and $R^\sfright_i$.
	
	\begin{algorithm}[h]
	\caption{Partition of Requests into Bundles for Both Single and Multiple-Vehicle Cases}
	\label{alg:partition}
		\begin{algorithmic}[1]
			\For{$i = 2, 4, 6, 8, \cdots$}
				\State wait until time $(i+1)F$ 
                \State let $R_i = \{\rho \in R: r_\rho \in [(i-1)F, iF)\}$ is as defined in the text				
			    \State partition $R_i$ into $R^\sfleft_i$ and $R^\sfright_i$ so as to minimize  
			    \begin{align*}
				    \begin{cases}
				    	\mst(R_{i-1}\cup R^\sfleft_{i})+\mst(R^\sfright_{i} \cup R_{i+1}) &  \textbf{if } k = 1\\[5pt]
				    	\cost(R^\sfleft_{i}|R_{i-1}) + \cost(R^\sfright_{i}|R_{i+1}) &	 \textbf{if }  k > 1
				    \end{cases}
			    \end{align*}
			    \State release the bundle $R'_{i-1}:=R^\sfright_{i-2} \cup R_{i-1} \cup R^\sfleft_i$  (at time $(i+1)F$)
			\EndFor
		\end{algorithmic}
	\end{algorithm}

	\begin{figure}[!htbp]
		\centering	
		\begin{subfigure}{\textwidth}
			\centering
		 	\includegraphics[width=0.9\textwidth]{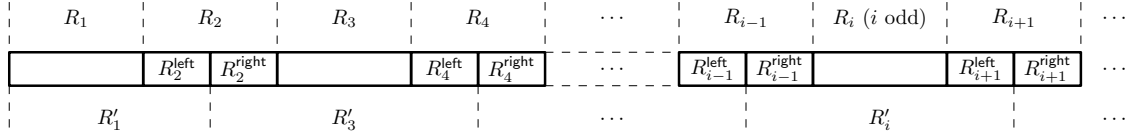}
		 	\caption{Illustration of relationships between  $R_i$'s, $R^\sfleft_i$'s, $R^\sfright_i$'s  and $R'_i$'s. For every even $i$, we have $R_i = R^\sfleft_i \uplus R^\sfright_i$. For every odd $i$, we have $R'_i = R^\sfright_{i-1}\uplus R_i \uplus R^\sfleft_{i+1}$.}
		\end{subfigure}\medskip
		
		\begin{subfigure}{\textwidth}
			\centering
			\includegraphics[width=0.9\textwidth]{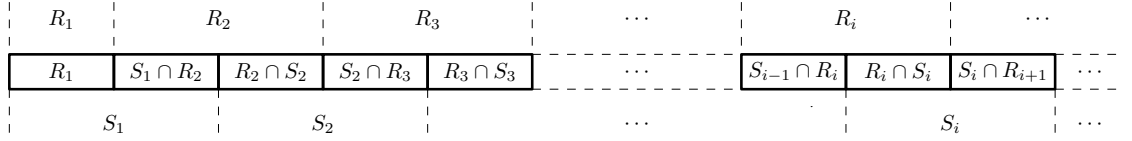}
			\caption{Illustration of relationships between $R_i$'s and $S_i$'s. For every $i$, we have $R_i \subseteq S_{i-1} \uplus S_i$ and $S_i \subseteq R_i \uplus R_{i+1}$.}
		\end{subfigure}\bigskip
		\caption{Illustration of sets used in the algorithm and analysis for online uncapacitated FDP on tree metrics.}
		\label{fig:sets}
	 \end{figure}
	  
    We briefly talk about the efficiency of the algorithm for $k = 1$. It is easy to see that the following simple strategy will find the partition $(R^\sfleft_i, R^\sfright_i)$ that minimizes $\mst(R_{i-1}\cup R^\sfleft_{i})+\mst(R^\sfright_{i} \cup R_{i+1})$. For every $\rho \in R_i$, if $d(\rho, \mst(R_{i-1})) \leq d(\rho, \mst(R_{i+1}))$, then we put $\rho$ in $R^\sfleft_i$. Otherwise we put $\rho$ in $R^\sfright_i$.  Here $d(\rho, \mst(R_{i-1}))$ ($d(\rho, \mst(R_{i+1}))$, resp.) is the shortest distance between $r_\rho$ and any vertex in $\mst(R_{i-1}))$ ($\mst(R_{i+1})$, resp.).

	Notice that starting from the depot $o$, serving all the requests in $R'_i$ and traveling back to the depot take time exactly $2\mst(R'_i)$.  Therefore, we can view each bundle $R'_i$, for an odd $i \geq 1$, as a job of size $2\mst(R'_i)$ released at time $(i+2)F$. We view the vehicle as a machine.  In the second step of the online algorithm, we then schedule the jobs (i.e., bundles) on the machine (i.e., the vehicle) in their order of releasing: Whenever the machine is idle and there are available jobs, we process the job that is released the earliest.  This finishes the description of the algorithm for the single vehicle case.

	\subsection{Analysis of Backlog of Jobs} We define the flow time of a job $R'_i$, for an odd $i \geq 1$, as its completion time minus its release time (which is $(i+2)F$). By the folklore result, if the total size of all jobs released in time $[t, t']$ is at most $t' - t + D$ for some $D$ and for every pair $t \leq t'$ of time points, then every job has flow time at most $D$ in the schedule. Therefore, it remains to prove that $D$ is small, which is stated in Lemma~\ref{lem:tree_travelingtimeintervals} later.
		
	Before describing and proving the lemma, we introduce a partition $(S_i)_{i \geq 1}$ of requests that depends on the offline optimum solution, and prove some helper lemmas.   Notice that in the offline optimum solution a trip can not serve two requests whose arrival times are more than $F$ apart: If that happens, then the flow time of the earlier request of the two is more than $F$. Thus,  a trip in the optimum solution serves either requests from a single set $R_i$ for some $i$, or requests from two sets $R_i$ and $R_{i+1}$ for some $i$. In both the cases, we put all the requests in the trip into the set $S_i$.  Therefore, $S_1, S_2, S_3, \cdots$ form a partition of $R$. Notice that $S_i \subseteq R_i \uplus R_{i+1}$ and $R_i \subseteq S_{i-1} \uplus S_i$ for every integer $i$. See Figure~\ref{fig:sets}(b) for an illustration of the relationship between $R_i$'s and $S_i$'s. Notice that $S_i$'s is only used in our analysis, as they depend on the offline optimum solution, which our algorithm does not know.

We need the following simple lemma and corollaries:
\begin{lemma}
	Let $X_1, X_2, X_3$ and $X_4$ be subsets of $V$ and $e \in E \cup V$. Then the following inequalities hold:
	\begin{align}
		\mst_e (X_1 \cup X_2) &\leq \mst_e(X_1) + \mst_e(X_2), \label{inequ:mst-simple-1} \\
		\mst_e(X_2) + \mst_e(X_1 \cup X_2 \cup X_3) &\leq \mst_e(X_1 \cup X_2) + \mst_e(X_2 \cup X_3), \label{inequ:mst-simple-2}\\
		\mst_e(X_1 \cup X_2 \cup X_3) + \mst_e(X_2 \cup X_3 \cup X_4) &\leq \mst_e(X_1 \cup X_2)  + \mst_e(X_2 \cup X_3) + \mst_e(X_3 \cup X_4). \label{inequ:mst-simple-3}
	\end{align}
\end{lemma}
\begin{proof}
	It suffices to prove the inequalities for $e = o$, since the other cases can be proved by changing $X_1, X_2, X_3$ and $X_4$ to $V_e \cap X_1, V_e \cap X_2, V_e \cap X_3$ and $V_e \cap X_4$ respectively.
	
	\eqref{inequ:mst-simple-1} is easy to see. For \eqref{inequ:mst-simple-2}, we focus on each edge $e' \in E$ and count the number of times $\ell(e')$ is considered on the left and right sides respectively.  Notice that $\ell(e')$ is counted in $\mst(Y)$ for some set $Y$ if and only if $V_{e'} \cap Y \neq \emptyset$.
	\begin{itemize}[itemsep=0pt]
		\item If $\ell(e')$ is counted $0$ times on the left side, then $V_{e'} \cap Y = \emptyset$ for every $Y \in \{X_1, X_2, X_3\}$. Clearly $\ell(e')$ is counted $0$ times on the right side.
		\item If $\ell(e')$ is counted in $\mst(X_1 \cup X_2 \cup X_3)$ but not in $\mst(X_2)$, then either $V_{e'} \cap X_1 \neq \emptyset$ or $V_{e'} \cap X_3 \neq \emptyset$. Then $\ell(e')$ is counted at least once on the right side.
		\item If $\ell(e')$ is counted in $\mst(X_2)$, then it is counted twice on both the left right sides.
	\end{itemize}\smallskip
	
	Similarly for the proof of \eqref{inequ:mst-simple-3}, we consider the number of times each $\ell(e')$ is counted on both sides:
	\begin{itemize}[itemsep=0pt]
		\item If $e'$ is counted 0 times on the left side, then $V_{e'} \cap Y = \emptyset$ for every $Y \in \{X_1, X_2, X_3, X_4\}$. Thus, $\ell(e')$ is counted $0$ times on the right side.
		\item Assume $\ell(e')$ is counted once on the left side. W.l.o.g assume it is counted in $\mst(X_1 \cup X_2 \cup X_3)$ but not in $\mst(X_2 \cup X_3 \cup X_4)$. Then clearly $e'$ is counted at least once on the right side. 
		\item Finally assume $\ell(e')$ is counted in both $\mst(X_1 \cup X_2 \cup X_3)$  and $\mst(X_2 \cup X_3 \cup X_4)$. If $V_{e'} \cap X_2 \neq \emptyset$ or $V_{e'} \cap X_3 \neq \emptyset$, then $\ell(e')$ is counted at least twice on the right side. Otherwise we have $V_{e'} \cap X_1 \neq \emptyset$ and $V_{e'} \cap X_4 \neq \emptyset$. In this case $\ell(e')$ is counted twice on the right side.  \qedhere
	\end{itemize}
\end{proof}

\begin{coro}
	\label{coro:R'i-less-than}
	For every odd $i \geq 1$ and $e \in E \cup V$, we have $\mst_e(R'_i) \leq \mst_e(R^\sfright_{i-1} \cup R_i) + \mst_e(R_i \cup R^\sfleft_{i+1}) - \mst_e(R_i)$.
\end{coro}
\begin{proof}
	Applying \eqref{inequ:mst-simple-2} with $X_1 = R^\sfright_{i-1}, X_2 = R_i$ and $X_3 = R^\sfleft_{i+1}$, and using $X_1 \cup X_2 \cup X_3 = R'_i$ proves the lemma. \footnote{Recall that in the notation $\mst_e$, we can view requests as vertices by ignoring their arrival times.}
\end{proof}

\begin{coro}
	\label{coro:less-than-Si-1-Si}
	For every integer $i \geq 1$ and $e \in E \cup V$, we have $\mst_e(S_{i-1} \cup R_i) + \mst_e(R_i \cup S_i) - \mst_e(R_i) \leq \mst_e(S_{i-1}) + \mst_e(S_i)$.
\end{coro}
\begin{proof}
	We apply \eqref{inequ:mst-simple-3} with $X_1 = R_{i-1} \cap S_{i-1}, X_2= S_{i-1} \cap R_i, X_3 = R_i \cap S_i$ and $X_4 = S_i \cap R_{i+1}$. The corollary follows by that $X_1 \cup X_2 = S_{i-1}$, $X_2 \cup X_3 = R_i$, $X_3 \cup X_4 = S_i$, $X_1 \cup X_2 \cup X_3 = S_{i-1} \cup R_i$ and $X_2 \cup X_3 \cup X_4 = R_i \cup S_i$.
\end{proof}

With the corollaries established, we now state  and prove the lemma that bounds the backlog of bundles.
\begin{lemma}
	\label{lem:tree_travelingtimeintervals}
	For any two odd positive integers $a \leq b$, we have 
	\begin{align*}
		2\sum_{i\in [a, b]: i\text{ odd}}\mst(R'_i) \leq (b-a + 4)F.
	\end{align*}
\end{lemma}
\begin{proof} 
	\begin{flalign}
		&& &\quad \sum_{i \in [a, b]:i\text{ odd}} \mst(R'_i) \nonumber\\
		&& &\leq \sum_{i \in [a, b]:i\text{ odd}} \big(\mst(R^\sfright_{i-1} \cup R_i) + \mst(R_i \cup R^\sfleft_{i+1}) - \mst(R_i)\big) &&\text{by Corollary~\ref{coro:R'i-less-than}}\nonumber\\
		&& &= \sum_{i \in (a, b): i\text{ even}}\Big(\mst(R_{i-1}\cup R^\sfleft_i) + \mst(R^\sfright_i \cup R_{i+1})\Big)- \sum_{i \in [a, b]: i\text{ odd}}\mst(R_i) \nonumber\\
		&& &\qquad + \mst(R^\sfright_{a-1} \cup R_a) + \mst(R_b\cup R^\sfleft_{b+1}) && \text{by reorganizing terms}\nonumber\\
		&& &\leq \sum_{i \in (a, b): i\text{ even}}\Big(\mst(R_{i-1}\cup S_{i-1}) + \mst(S_i \cup R_{i+1})\Big)- \sum_{i \in [a, b]: i\text{ odd}}\mst(R_i) \nonumber\\
		&& &\qquad + \mst(S_{a-2}\cap R_{a-1}) + \mst(S_{a-1} \cup R_a)  + \mst(R_b\cup S_b) + \mst(R_{b+1} \cap S_{b+1}) \label{inequ:tree-1}\\
		&& &= \sum_{i \in [a, b]:i\text{ odd}} \big(\mst(S_{i-1} \cup R_i) + \mst(R_i \cup S_i) - \mst(R_i)\big)\nonumber\\
		&& &\qquad +\mst(S_{a-2} \cap R_{a-1}) + \mst(R_{b+1} \cap S_{b+1})&&\text{by reorganizing terms}\nonumber\\
		&& &\leq \sum_{i \in [a, b]:i\text{ odd}}\left(\mst(S_{i-1}) + \mst(S_i)\right) + \mst(S_{a-2} \cap R_{a-1}) + \mst(R_{b+1} \cap S_{b+1}) &&\text{by Corollary \ref{coro:less-than-Si-1-Si}}\nonumber\\
		&& &=\mst(S_{a-2} \cap R_{a-1}) +  \sum_{i = a-1}^b\mst(S_i) + \mst(R_{b+1} \cap S_{b+1}).\nonumber
	\end{flalign}
	It remains to argue about \eqref{inequ:tree-1}.  Notice that $S_{i-1} \cap R_i$ and $R_i \cap S_i$ form a partition of $R_i$.  By the way we obtain $R^\sfleft_i$ and $R^\sfright_i$, we have $\mst(R_{i-1} \cup R^\sfleft_i) + \mst(R^\sfright_i \cup R_{i+1}) \leq \mst(R_{i-1} \cup (S_{i-1} \cap R_i)) + \mst((R_i \cap S_i) \cup R_{i+1}) = \mst(R_{i-1} \cup S_{i-1}) + \mst(S_i \cup R_{i+1})$.  Then $R^\sfright_{a-1} \cup R_a \subseteq R_{a-1} \cup R_a = (S_{a-2} \cap R_{a-1}) \cup S_{a-1} \cup R_a$ and $R_b \cup R^\sfleft_{b+1} \subseteq R_b \cup R_{b+1} = R_b \cup S_b \cup (R_{b+1}\cap S_{b+1})$.  Combining the facts with \eqref{inequ:mst-simple-1} gives \eqref{inequ:tree-1}.
	
	Now we prove that $2\left(\mst(S_{a-2} \cap R_{a-1}) + \sum_{i = a-1}^{b}\mst(S_i) + \mst(R_{b+1} \cap S_{b+1})\right)\leq (b-a+4)F$, which finishes the proof of the lemma.  We define $\calS= \{(S_{a-2} \cap R_{a-1}), S_{a-1}, S_a, S_{a+1}, \cdots, S_b, (R_{b+1} \cap S_{b+1})\}$ for convenience. Then $\calS$ forms a partition of $Q := R_{a-1} \cup R_a \cup R_{a+1} \cup \cdots \cup R_{b+1}$.  Focus on the trips in the optimum solution that serve at least one request in $Q$.  By the definition of $S_i$'s, each such trip can not serve requests from two different sets in $\calS$.  Moreover, such a trip can not start before time $(a-2)F$ since all requests in $Q$ arrive no earlier than $(a-2)F$. It can not end after $(b+1)F + F = (b+2)F$ since all requests in $Q$ arrive before $(b+1)F$ and the maximum flow time of the optimum solution is at most $F$.  Therefore, $2\sum_{S \in \calS}\mst(S)$ is at most the total length of these trips, which is at most $(b+2)F  - (a-2)F = (b-a+4)F$. 
\end{proof}


By Lemma~\ref{lem:tree_travelingtimeintervals}, the total size of jobs (i.e, bundles) released in $[t, t']$ is at most $(t' - t) + 4F$ for any $t \leq t'$, using the language of the scheduling setting. Therefore, the greedy scheduling algorithm produce a schedule with maximum flow time at most $4F$.  For an odd $i$, $R'_i$ is released at time $(i+2)F$, and thus all requests in $R'_i$ are served by time $(i+6)F$. Since requests in $R'_i$ arrive no earlier than $(i-2)F$, the maximum flow time of all requests is at most $8F$. This finishes the proof of Theorem~\ref{thm:ALG-on-trees} for the case $k = 1$. 

\section{Online Uncapacitated Multiple-Vehicle FDP on Tree Metrics}
\label{sec:tree-multi}
Now we move on to the case of general $k$ for online FDP on trees. As we mentioned, the main differences between the algorithm and the one for the single-vehicle case are: We use a different criteria to break $R_i$ for an even $i \geq 2$ into $R^\sfleft_i$ and $R^\sfright_i$, and we break a bundle $R'_i$ for an odd $i \geq 1$ into groups and treat each group as a job (instead of treating the whole bundle as a job). 

\subsection{Partitioning $R$ into Bundles}
As before we also generate the set $(R'_{i})_{i \geq 1, i\text{ is odd}}$ in the first step, which is also described in Algorithm~\ref{alg:partition}. To break $R_i$ into $R^\sfleft_i$ and $R^\sfright_i$, we use  the function $\cost(R^\sfleft_i|R_{i-1}) + \cost(R^\sfright_i|R_{i+1})$.   We define the relevant notations now.

Given a real number $F' > 0$ and a set $X$ of requests, we define
\begin{align*}
	c_{F'}(X, e) = \ceil{\frac{\mst_e(X)}{F'}}, \forall e \in E, \text{ and } \cost_{F'}(X) = 2\sum_{e \in E} c_{F'}(X, e)\ell(e).
\end{align*}
Most of the time we shall use the definitions with $F' = F$. Therefore we simply use $c(X, e)$ and $\cost(X)$  to denote $c_{F}(X, e)$ and $\cost(X)$.

Then, for two sets $X$ and $X'$ of requests, we define
		\begin{align*}
			c(X, e | X') &= \begin{cases}
			\ceil{\frac{\mathsf{\mst}_e(X)}{F}} = c(X, e)	& \text{if } V_e \cap X' = \emptyset\\[5pt]
			\floor{\frac{\mathsf{\mst}_e(X' \cup X) - \mathsf{\mst}_e(X')}{F}} & \text{if } V_e \cap X' \neq \emptyset
			\end{cases}, \quad 
			\forall e\in E, \\
			\text{ and } \cost(X | X') &= 2\sum_{e \in E} c(X, e|X')\ell(e).
		\end{align*}
By the definition, we have $c(X, e|X') = c(X \setminus X', e|X')$, and $\cost(X|X') = \cost(X \setminus X' | X')$.  

We now elaborate more on the definitions. Unlike the single-vehicle case, we need to break each bundle $R'_i$ created into many groups, to make sure that each group can be served in time $O(F)$.  Then an edge $e \in E$ in $\mst(R'_i)$ may need to be used by many trips to satisfy the property. Then $c(X, e)$ gives a lower bound on the number of bi-directional traverses of $e$ needed to serve $X$ in the offline optimum solution. Thus $\cost(X)$ gives the total time needed to serve $X$. Notice if a trip uses an edge $e$, it uses $e$ twice, hence the factor of $2$.

For $c(X, e | X')$, one can think of it as $c(X \cup X', e) - c(X', e) = \ceil{\frac{\mst_e(X \cup X')}{F}} - \ceil{\frac{\mst_e(X')}{F}}$, which is at least $\floor{\frac{\mst_e(X \cup X')}{F} - \frac{\mst_e(X')}{F}}$.  That is, this is the extra number of traverses of $e$ needed if we grow the set of request positions from $X'$ to $X \cup X'$.  In the actual definition, we use the lower bound instead if $V_e \cap X' \neq \emptyset$. 

\subsection{Upper Bound on Costs of Bundles}

The main goal of the section is to prove the following theorem:
\begin{theorem}
	\label{thm:tree-k-backlog}
	For every two odd positive integers $a \leq b$, we have 
	\begin{align*}
		\sum_{i \in [a, b]}\cost_{3F}(R'_i) \leq k(b - a + 4)F.
	\end{align*}
\end{theorem}

We prove the following three inequalities, which will imply the theorem:
\begin{align}
		\sum_{i \in [a, b]}\cost_{3F}(R'_i)
		&\leq\sum_{i \in [a, b]:i\text{ odd}}\Big( \cost(R_i) + \cost(R_{i-1}^\sfright|R_i) + \cost(R_{i+1}^\sfleft|R_i) \Big) \label{inequ:tree-2}\\
		&\leq \cost(S_{a-2} \cap R_{a-1}) + \sum_{i = a-1}^{b}\cost(S_i) + \cost(R_{b+1} \cap S_{b+1}) \label{inequ:tree-3}\\
&\leq k(b-a+4)F. \label{inequ:tree-4}
\end{align}

We first prove \eqref{inequ:tree-4}, by showing that the $\cost()$ function indeed capture the length of trips in optimum solution. 
\begin{claim}
\label{lem:candcost}
	 Let $R' \subseteq R$, and $\calP$ be the set of trips in the offline optimum solution that serve at least one request in $R'$. Then every $e \in E$ is used by at least $c(R', e)$ trips in $\calP$, implying that the total cost of $\calP$ is at least $\cost(R')$. 
\end{claim}	
\begin{proof}
	Focus on each edge $e \in E$, and assume $e = (u, v)$, where $v$ is the child vertex.   If $V_e \cap R' = \emptyset$ then $c(R', e) = 0$ and the statement holds trivially.  If $V_e \cap R' \neq \emptyset$ and $\mst_e(R') \leq F$, then $c(R', e) = 1$ and at least $1$ trip in $\calP$ uses $e$ and the claim also holds. So, from now on, we assume $\mst_e(R') > F$. 
	
	Focus on a trip $P \in \calP$ that uses $e$. For convenience, we treat $P$ as the sub-tree of edges used by $P$, without double-counting each edge. 
	The total length of descendant edges of $e$ in $P$ (excluding $e$ itself) is at most $F - d(o, v)$. This holds since $P$ contains edges with a total length of at most $F$, and it contains the edges from $o$ to $v$.   On the other hand, all the descendant edges of $v$ in $\mst_e(R')$ should be contained in $\calP$. The total length of these edges is $\mst_e(R') - d(o, v)$. Therefore, the number of trips that use $e$ is at least $\ceil{\frac{\mst_e(R') - d(o, v)}{F - d(o, v)}} \geq \ceil{\frac{\mst_e(R')}{F}} = c(R', e)$, where the inequality comes from that $\mst_e(R') > F \geq d(o, v)$.  The lemma holds as each trip that traverses $e$ does this twice. 
\end{proof}

\begin{proof}[Proof of \eqref{inequ:tree-4}]The argument is similar to that in the last paragraph inside the proof of Lemma~\ref{lem:tree_travelingtimeintervals}, except now we use $\cost(S)$, instead of $2\mst(S)$, to lower bound the length of trips for a set $S \in \calS$ of requests, and there are  $k \geq 2$ vehicles. \end{proof}

Then we turn to \eqref{inequ:tree-3}. We first show some simple inequalities about the $\cost$ function. 
\begin{lemma}
	For any $X_1, X_2, X_3, X_4 \subseteq V$, we have 
	\begin{align}
		\cost(X_1 \cup X_2 | X_3) &\leq \cost(X_1) + \cost(X_2 | X_3), \label{inequ:tree-5}\\
		\cost(X_2 \cup X_3) + \cost(X_1|X_2 \cup X_3) + \cost(X_4|X_2 \cup X_3) &\leq \cost(X_1 \cup X_2) + \cost(X_3 \cup X_4). \label{inequ:tree-6}
	\end{align}
\end{lemma}

\begin{proof}
	For notational convenience, we use $X_{s}$ for any string $s$ with alphabet $\{1,2,3,4\}$ to denote $\bigcup_{i \in s}X_i$. For example, $X_{123} = X_1 \cup X_2 \cup X_3$. 
		It suffices for us to prove that for every $e \in E$, we have 
		\begin{align}
			c(X_{12}, e | X_3) &\leq c(X_1, e) + c(X_2, e | X_3), \label{inequ:tree-7}\\
			c(X_{23}, e) + c(X_1, e|X_{23}) + c(X_4, e|X_{23})&\leq c(X_{12}, e) + c(X_{34}, e). \label{inequ:tree-8}
		\end{align} 
	Throughout the proof, we shall use the following simple inequalities for two reals $a, b$: $\floor{a + b} \leq \floor{a} + \ceil{b} \leq \ceil{a + b} \leq \ceil{a}+\ceil{b}$.

	First consider the \eqref{inequ:tree-7}. If $X_3 \cap V_e = \emptyset$, then it holds since its left side is $c(X_{12}, e) = \ceil{\frac{\mst_e(X_{12})}{F}} \leq \ceil{\frac{\mst_e(X_1) + \mst_e(X_2)}{F}} \leq \ceil{\frac{\mst_e(X_1)}{F}} + \ceil{\frac{\mst_e(X_2)}{F}} = c(X_1 ,e) + c(X_2, e) = c(X_1 ,e) + c(X_2, e|X_3)$, which is the right side. If $X_3 \cap V_e \neq \emptyset$, then the left side is $\floor{\frac{\mst_e(X_{123}) - \mst_e(X_3)}{F}}\leq \floor{\frac{\mst_e(X_1) + \mst_e(X_{23}) - \mst_e(X_3)}{F}} \leq \ceil{\frac{\mst_e(X_1)}{F}} + \floor{\frac{\mst_e(X_{23}) - \mst_e(X_3)}{F}} = c(X_1, e) + c(X_2, e | X_3)$, which is the right side of the inequality. 
		
	Then we move on to prove \eqref{inequ:tree-8}. Suppose $V_e \cap X_{23}= \emptyset$.  Then we have $c(X_{23}, e) = 0$ and $c(X_1, e|X_{23}) = c(X_1, e) = c(X_{12},e)$ and $c(X_4, e|X_{23})= c(X_4, e)=c(X_{34},e)$.   \eqref{inequ:tree-8} holds with equality. So we assume $V_e \cap X_{23}\neq \emptyset$. Then,  
	\begin{align*}
		&\quad c(X_{23}, e) + c(X_1, e|X_{23}) + c(X_4, e|X_{23}) \\
		&= \ceil{\frac{\mst_e(X_{23})}{F}} + \floor{\frac{\mst_e(X_{123}) - \mst_e(X_{23})}{F}} + \floor{\frac{\mst_e(X_{234}) - \mst_e(X_{23})}{F}}\\
		&\leq \ceil{\frac{\mst_e(X_2)}{F}} + \ceil{\frac{\mst_e(X_3)}{F}}+ \floor{\frac{\mst_e(X_{123}) - \mst_e(X_{23})}{F}} + \floor{\frac{\mst_e(X_{234}) - \mst_e(X_{23})}{F}}\\
		&\leq \ceil{\frac{\mst_e(X_2)+\mst_e(X_{123}) - \mst_e(X_{23})}{F}} + \ceil{\frac{\mst_e(X_3)+\mst_e(X_{234}) - \mst_e(X_{23})}{F}}\\
		&\leq \ceil{\frac{\mst_e(X_{12})}{F}} + \ceil{\frac{\mst_e(X_{34})}{F}} = c(X_{12}, e) + c(X_{34}, e). && 
	\end{align*}	
	The two equalities are by definitions. The first and the last inequalities follow \eqref{inequ:mst-simple-1} and \eqref{inequ:mst-simple-2} respectively.
\end{proof}

Let $S_i$'s be defined in the same way as in the single-vehicle case: all the requests in a trip containing only requests from $R_i$, or requests from $R_i$ and $R_{i+1}$, are put into $S_i$. 
\begin{coro}
    \label{coro:break-cost}
    	For any $i$, we have $\cost(R_i) + \cost(R_{i-1} \cap S_{i-1}|R_i) + \cost(S_i\cap R_{i+1}|R_i) \leq \cost(S_{i-1}) + \cost(S_i)$.
\end{coro}
\begin{proof}
	The corollary follows from \eqref{inequ:tree-6} by setting $X_1 = R_{i-1} \cap S_{i-1}, X_2 = S_{i-1} \cap R_i, X_3 = R_i \cap S_i$ and $X_4 = S_i \cap R_{i+1}$. 
\end{proof}

Recall that in our algorithm, for every even $i$, we break $R_i$ into $R^\sfleft_{i}$ and $R^\sfright_{i}$ so as to minimize $\cost(R^\sfleft_{i}|R_{i-1}) + \cost(R^\sfright_{i}|R_{i+1})$.  
\begin{proof}[Proof of  \eqref{inequ:tree-3}]
	\begin{flalign}
		&& &\quad \sum_{i \in [a, b]: i\text{ odd}}\Big(\cost(R_i) + \cost(R_{i-1}^\sfright|R_i) + \cost(R_{i+1}^\sfleft|R_i)\Big) \nonumber\\
		&& &=\sum_{i \in [a, b]: i\text{ odd}}\cost(R_i) + \sum_{i \in (a, b): i\text{ even}}\Big(\cost(R_i^\sfleft | R_{i-1}) + \cost(R_i^\sfright | R_{i+1}) \Big)&&\nonumber\\
		&& &\qquad + \quad\cost(R_{a-1}^\sfright|R_a) + \cost(R_{b+1}^\sfleft|R_b) &&\text{by reorganizing terms}\nonumber\\
		&& &\leq\sum_{i \in [a, b]: i\text{ odd}}\cost(R_i) + \sum_{i \in (a, b): i\text{ even}}\Big(\cost(S_{i-1} \cap R_i | R_{i-1}) + \cost(R_i \cap S_i | R_{i+1}) \Big) &&\nonumber\\
		&& &\qquad + \quad\cost(S_{a-2} \cap R_{a-1}) + \cost(R_{a-1} \cap S_{a-1}|R_a) &&\nonumber\\
		&& &\qquad + \quad \cost(S_b \cap R_{b+1}|R_b) + \cost(R_{b+1} \cap S_{b+1})  \label{inequ:tree-9} &&\\
		&& &= \sum_{i \in [a, b]: i\text{ odd}}\big(\cost(R_i) + \cost(R_{i-1} \cap S_{i-1}|R_i) + \cost(S_{i} \cap R_{i+1}| R_i)\big) && \nonumber\\
		&& &\qquad  + \quad \cost(S_{a-2} \cap R_{a-1}) + \cost(R_{b+1}\cap S_{b+1}) && \text{by reorganizing terms}\nonumber\\
		&& &\leq \sum_{i \in [a, b]: i \text{ odd}} \Big(\cost(S_{i-1}) + \cost(S_i)\Big) + \cost(S_{a-2} \cap R_{a-1}) + \cost(R_{b+1}\cap S_{b+1})  &&\text{by Corollary~\ref{coro:break-cost}}\nonumber\\
		&& &= \cost(S_{a-2} \cap R_{a-1}) + \sum_{i = a-1}^{b}\cost(S_i) + \cost(R_{b+1} \cap S_{b+1}). &&\nonumber
	\end{flalign}
	We need to elaborate more on \eqref{inequ:tree-9}.  Notice that $S_{i-1} \cap R_i$ and $R_i \cap S_i$ form a partition of $R_i$. By the way we choose $R^\sfleft_i$ and $R^\sfright_i$, we have $\cost(R^\sfleft_i|R_{i-1}) + \cost(R^\sfright_i|R_{i+1}) \leq \cost(S_{i-1} \cap R_i|R_{i-1}) + \cost(R_i \cap S_i | R_{i+1})$.  Then $\cost(R^\sfright_{a-1}|R_a) \leq \cost(R_{a-1}|R_a) \leq \cost(S_{a-2} \cap R_{a-1}) + \cost(R_{a-1} \cap S_{a-1}|R_a)$, follows from \eqref{inequ:tree-5} by letting $X_1 = S_{a-2} \cap R_{a-1}, X_2 = R_{a-1} \cap S_{a-1}$ and $X_3 = R_a$. Similarly, we have $\cost(R^\sfleft_{b+1}|R_b) \leq \cost(S_b \cap R_{b+1}|R_b) + \cost(R_{b+1} \cap S_{b+1})$. 	
\end{proof}

Now we proceed to the proof of \eqref{inequ:tree-2}. Focus on an odd $i \in [a, b]$, and any edge $e \in E$.  We need to prove 
\begin{align}
	c(R_i, e) + c(R_{i-1}^\sfright, e|R_i) + c(R_{i+1}^\sfleft, e|R_i) \geq c_{3F}(R'_i, e)=\ceil{\frac{\mst_e(R'_i)}{3F}}. \label{inequ:tree-10}
\end{align}
\eqref{inequ:tree-2} is proved by multiplying \eqref{inequ:tree-10} by a factor of $2\ell(e)$, and summing up the resulting inequality over all odd integers $i \in [a, b]$ and $e \in E$.

 	Suppose $V_e \cap R_i= \emptyset$. Then the left side of  \eqref{inequ:tree-10} is $\ceil{\frac{\mst_e(R^\sfright_{i-1})}{F}} + \ceil{\frac{\mst_e(R^\sfleft_{i+1})}{F}} \geq \ceil{\frac{\mst_e(R^\sfright_{i-1} \cup R_i \cup R^\sfleft_{i+1})}{F}} = \ceil{\frac{\mst_e(R'_i)}{F}} \geq \ceil{\frac{\mst_e(R'_i)}{3F}}$. Now assume $V_e \cap R_i \neq \emptyset$.  Then the left side is $\ceil{\frac{\mst_e(R_i)}{F}} + \floor{\frac{\mst_e(R^\sfright_{i-1}\cup R_i)-\mst_e(R_i)}{F}} + \floor{\frac{\mst_e(R^\sfleft_{i+1}\cup R_i) - \mst_e(R_i)}{F}} \geq \floor{\frac{\mst_e(R^\sfright_{i-1} \cup R_i) + \mst_e(R^\sfleft_{i+1} \cup R_i)-\mst_e(R_i) }{F}}-1 \geq \floor{\frac{\mst_e(R^\sfright_{i-1} \cup R_i \cup R^\sfleft_{i+1})}{F}}-1 = \floor{\frac{\mst_e(R'_i)}{F}}-1$. On the other hand, the left side is at least $1$.  So, the left side is at least $\max\Big\{1, \floor{\frac{\mst_e(R'_i)}{F}}-1\Big\} \geq \ceil{\frac{\mst_e(R'_i)}{3F}}$, where we used that $\max\{1, \floor{a}-1\} \geq \ceil{\frac a3}$ for every $a \geq 0$. 
 	
 	So we have finished the proof of Theorem~\ref{thm:tree-k-backlog}.
 	
\subsection{Breaking Bundles into Groups}
	In this section, we break each bundle $R'_i$ into many groups $\calQ_i$ as in the following lemma. 
  \begin{lemma}
	  	\label{lemma:tree-general-k-backlog}
	 	For every odd integer $i \geq 1$, we can efficiently find a partition $\calQ_i$ of $R'_i$ such that $\mst(Q) \leq 8F$ for every $Q \in \calQ_i$, and $2\sum_{Q \in \calQ_i}\mst(Q) \leq \cost_{3F}(R'_i)$.
  \end{lemma}

  \begin{proof}
  		Within this proof, we need to change $T$ to a binary tree by replacing each internal vertex $v$ with at least three children with a binary tree. For every such vertex $v$ with $d_v \geq 3$ children, we replace the star containing $v$ and its children by a gadget, which is a complete binary tree with $d_v$ leaves. We then identify the root of the gadget with $v$, and the leaves with the $d_v$ children of $v$. In the gadget, the length of an edge incident to a child $u$ of $v$ is set to $\ell(u, v)$, and the length of an edge not incident to a child is set to $0$. It is easy to see that this transformation does not change the instance, except now we have edges of length $0$. After this step, every internal vertex of $T$ has degree exactly $2$.  
  		
 		For any vertex $v$ in $T$ and a set $R'$ of requests, we define $\mst'_v(R')$ to be cost of the minimum spanning tree containing $v$ and $V_v \cap R'$, where $V_v$ is the set of descendant vertices of $v$ in $T$ (including $v$ itself). Notice an important difference between the definition of $\mst'_v(R')$ and $\mst_v(R')$ for an edge $e$: for $\mst'_v(R')$ the tree does not need to contain the depot $o$ and thus the quantity does not count the total length $d(o, v)$ of edges from $o$ to $v$.   Similarly, for an edge $e = (u, v)$ with $v$ being the child end-vertex, we use $\mst'_e(R')$ to denote $\mst'_v(R')$. Also, sometimes we use $\mst'_v(R')$ to denote the tree achieving the cost $\mst'_v(R')$.
 		
 		The following is the pseudo-code for constructing $\calQ_i$: 
 		\begin{algorithm}[H]
 		\begin{algorithmic}[1]
 			\State $R' \gets R'_i$, $\calQ_i \gets \emptyset$. 
 			\While{$R' \neq \emptyset$}
 				\State choose a lowest vertex $v$ such that $\mst'_v(R') \geq 3F$; if no vertices $v$ satisfy the condition, let $v = o$
 				\State $Q \gets \{\rho \in R': v_\rho \in V_v\}, \calQ\gets \calQ \cup \{Q\}, R' \gets R' \setminus Q$. 
 			\EndWhile
 		\end{algorithmic}
 		\end{algorithm}
 		
		It is easy to see that $\calQ$ form a partition of $R'_i$.  Focus on any iteration of the loop, and let $v$ be the vertex chosen in the iteration. We argue that we have $\mst'_v(R') + d(o, v) \leq 8F$, implying that the set $Q$ added in the iteration has $\mst(Q) = \mst'_v(R') + d(o, v) \leq 8F$.  To see this, assume the two children of $v$ in $T$ are $u$ and $u'$. (If $v$ is a leaf the statement is trivial.)  By our choice of $v$ we have $\mst'_u(R') < 3F$ and $\mst'_{u'}(R') < 3F$. Then we have $\mst_v(R') \leq \mst'_u(R') + \mst'_{u'}(R') + d(o, v) + d(v, u) + d(v, u') \leq 3F+3F + d(o, u) + d(o, u') \leq 6F + F + F = 8F$. The first inequality may not hold with equality since $(v, u)$ or $(v, u')$ may not be in $\mst'_v(R')$. 
  	
  	It remains to prove that for every edge $e  \in E$, at most $\ceil{\frac{\mst'_e(R'_i)}{3F}} \leq \ceil{\frac{\mst_e(R'_i)}{3F}} = c_{3F}(R'_i, e)$ groups $Q \in \calQ_i$ use the edge $e$.  Consider any edge $e$ and assume the group $Q$ is constructed in an iteration in which the vertex we choose is $v$. If $v \in V_e$, then $\mst'_e(R')$ is reduced by at least $3F$ in the iteration since all the edges in $\mst'_v(R')$ are removed from $\mst'_e(R')$. Otherwise, $v$ must be an ancestor of the parent end-vertex of $e$. In this case, $\mst'_e(R')$ becomes $\emptyset$ and thus this is the last time $e$ is used.  Moreover once $\mst'_e(R')$ becomes $0$, there are no requests in $V_e$.  Therefore, $e$ is used  at most $\ceil{\frac{\mst'_e(R'_i)}{3F}}$ times. 
   \end{proof}
  
\subsection{Scheduling Groups Greedily}

	Once we partitioned each bundle $R'_i$ into many groups $\calQ_i$, we can then treat the groups as jobs and the $k$ vehicles as $k$ machines. Each group $Q \in \calQ_i$ can be viewed as a job of size $2\mst(Q)$ that is released at time $(i+2)F$.   We need to analyze the maximum flow time achieved using the FIFO algorithm. 
	
	Now we define the goal more formally in the language of the scheduling problem. A job $j$ arrives time $r_j$, and upon its arrival, we know its processing time $p_j$. We need to schedule the jobs on the $k$ machines non-preemptively so as to minimize the maximum flow time over all jobs, where the flow time of a job is its completion time minus its release time. We need to design an online algorithm: Once we started processing a job on a machine, we need to complete the job on the machine.  Needless to say, the decisions made at or before time $t$ can only depends on the jobs arrived at or before time $t$.  
    
    We consider the simple FIFO (first-in-first-out) algorithm: whenever there is an idle machine and an available job (an arrived job that is not being processed), we process the available job with the earliest releasing time on the machine. The following simple lemma is implicit in the analysis of FIFO on multiple machine scheduling to minimize the maximum flow time:
	\begin{lemma}
		\label{lemma:scheduling-k-machines}
		Let $P$ be the maximum processing times of all jobs and $D \geq 0$. 
		Assume for any two time points $a \leq b$, the total size of jobs released in time $[a, b]$ is at most $k(b - a) + D$, then the maximum flow time achieved by FIFO is at most $\frac{D}{k} + \frac{2(k-1)P}{k}$. 
	\end{lemma}
	Notice that when $k = 1$, the upper bound is simply $D$, and it is a folklore that the FIFO algorithm is optimum. 
	
	\begin{proof}
		Focus on a job $j$ and let $s_j$ be its starting time in the schedule produced by the greedy algorithm. 	Let $t \leq s_j$ be the latest time such that some machine is idle in $(t - \epsilon, t)$ for some positive $\epsilon$.  (It is possible that $t = 0$.) Notice that we have $t \leq r_j$ since otherwise $j$ could be started on that machine at time $t - \epsilon$.

		All the $k$ machines are busy in time $(t, s_j)$.   All but at most $k-1$ jobs that are processed fully or partially in $(t, s_j)$ are released in $[t, r_j]$. The total size of the jobs (including the at most $k-1$ jobs) is at most $(k-1)P  + k(r_j - t) +D$, by the conditions of the lemma. Therefore we have $k(s_j - t) + p_j \leq (k-1)P  + k(r_j - t) +D$, which is $s_j + \frac{p_j}{k} \leq \frac{(k-1)P}{k} + r_j + \frac{D}{k}$. This implies the flow time of $j$ is $s_j + p_j - r_j \leq \frac{(k-1)P}{k} + \frac{D}{k} + \frac{(k-1)}{k}p_j \leq \frac{D}{k} + \frac{2(k-1)P}{k}$.
	\end{proof}

	We then use Lemma~\ref{lemma:scheduling-k-machines} to bound the maximum flow time for the food delivery problem. By Lemma~\ref{lemma:tree-general-k-backlog}, all jobs (which correspond to groups) have size at most $16F$. By Lemma~\ref{lemma:tree-general-k-backlog} and Theorem~\ref{thm:tree-k-backlog}, the total size of jobs released at any interval $[t, t']$ is at most $(t' - t)k + 4kF$. Therefore, using the greedy algorithm and Lemma~\ref{lemma:scheduling-k-machines}, every job will have flow time at most $\frac{4kF}{k} + 2 \times 8F = 20F$. Notice that requests in $R'_i$ has arrival time at least $(i-2)F$ and is  released at time $(i+2)F$. So, every request has flow time at most $20F + 4F = 24F$. This finishes the proof of Theorem~\ref{thm:ALG-on-trees} for general $k$, except for the proof of the running time of Algorithm~\ref{alg:partition}, which we argue next.
	
	Assuming all lengths are integers. Then it is easy to design a pseudo-polynomial time algorithm to find a partition $(R^\sfleft_i, R^\sfright_i)$ of $R_i$ to minimize $\cost(R^\sfleft_i|R_{i-1}) + \cost(R^\sfright_i|R_{i+1})$ using dynamic programming. Consider any two edges $e'$ and $e$ such that $e'$ is an ancestor of $e$. Then given $\mst_e(R^\sfleft_i \cup R_{i-1}) - \mst_e(R_{i-1})$, $c(R^\sfleft_i, e'|R_{i-1})$ does not depend on the set $\big\{\rho \in R^\sfleft_i: v_\rho \in V_e\big\}$. Similarly, given $\mst_e(R^\sfright_i \cup R_{i+1}) - \mst_e(R_{i+1})$, $c(R^\sfright_i, e'|R_{i+1})$ does not depend on the set $\big\{\rho \in R^\sfright_i: v_\rho \in V_e\big\}$. Therefore, we can design a dynamic programming where for each $e$ and two integers $M_\sfleft$ and $M_\sfright$, we have a cell indicating if there is a partition $(R^\sfleft_i, R^\sfright_i)$ of $R_i$ such that $\mst_e(R^\sfleft_i \cup R_{i-1}) - \mst_e(R_{i-1}) = M_\sfleft$ and $\mst_e(R^\sfright_i \cup R_{i+1}) - \mst_e(R_{i+1}) = M_\sfright$. 
	To compute the value of a cell $(e, M_\sfleft, M_\sfright)$, we look at all cell $(e', M'_\sfleft, M'_\sfright)$s with $e'$ being a \emph{child} of $e$, and check if there's a way to combine the $M'_\sfleft$s (resp. $M'_\sfright$s) to get $M_\sfleft$ (resp. $M_\sfright$). (To make this step fast we may assume the input tree is a binary tree: this can be done by adding zero-length dummy edges and vertices)
	To make the algorithm truly polynomial, we can round all edge lengths to integer multiples of $\epsilon F/(n|R|)$ for any small constant $\epsilon > 0$. Since every one of the $n$ edges is used by at most $|R|$ requests, the total error incurred due to the rounding is at most $\epsilon F/(n|R|)\cdot n|R| = \epsilon F$, which can be ignored.

    \section{Online FDP on General Metrics with Speeding}
\label{sec:speeding}
We study online FDP on general metrics. As suggested by (\ref{thm:LB}b), it is impossible to get an $o(n)$-competitive ratio for the problem. To circumvent the pessimistic instances, we resort to the speed-augmentation model and propose an $O(1/\epsilon)$-competitive algorithm with $(1+\epsilon)$-speeding. It requires exponential computational power to find the best TSP/CVRP tour. However, if we only allow polynomial running time, then it suffices for us to use speeding factor $\alpha_\TSP + \epsilon$ or $\alpha_\CVRP + \epsilon$. The speeding factors are tight due to (\ref{thm:LB-speeding}b) and (\ref{thm:LB-speeding}c). 



A simple but crucial lemma we use is the following: 
\begin{lemma}
    \label{lem:bound-tsp}
    Let $a \leq b$ and let $R'$ be the set of requests that arrive in time $[a, b]$. Then the best CVRP tour for $R'$ with capacity $c$ has length at most $k(b - a + 2F)$.
\end{lemma}
\begin{proof}
	Focus on the trips in the optimum solution that serve at least one request in $R'$.  For each such trip, we consider the last request $\rho \in R'$ it serves, and change the trip so that it returns to $o$ directly using the shortest path after serving $\rho$. So a trip in the set can not start before $a$, since otherwise it can not serve a request in $R'$. Every request in $R'$ has flow time at most $F$ and thus serving time at most $b + F$. Since it takes at most $F$ units time to travel from a vertex back to $o$ (otherwise the vertex is useless), the completion time of each trip is at most $b + 2F$. Therefore, all these trips have starting and completion times in $[a, b+2F]$. Since there are $k$ vehicles, the total length of trips is at most $k(b-a + 2F)$. Moreover, each trip serves at most $c$ requests and so they together give a valid solution for the CVRP problem with capacity $c$. 
\end{proof}
	
%

If there was only one vehicle, we can immediately obtain a $(1+\epsilon)$-speeding $O(1/\epsilon)$-competitive algorithm: For $\gamma = 2/\epsilon$, we break the time horizon into intervals of length $\gamma F$, and construct a trip of length $(\gamma+2)F$ for requests that arrive in each interval. With $(\gamma + 2)/\gamma = 1+\epsilon$ speeding, the trip can be completed within time $\gamma F$. 

The situation becomes slightly more complicated when we have multiple vehicles. Now the best TSP/CVRP tour length is bounded by $k(\gamma + 2)F$, we need to break it into $k$ sub tours in a balanced way. 
Fortunately, we have the fact that the distance from $o$ to any $v$ in the metric is at most $F$ in the TSP tour. We keep tracing the TSP tour until the first time the tour's length is larger than $O(F/\epsilon)$ or the TSP tour ends, then we make it as one sub tour. The number of these sub tours is at most $k$. 
All vehicles can come back before $(i+1)\gamma F$ with a certain speeding factor. 
The algorithm is given in Algorithm~\ref{alg:multi-general}.\footnote{By the definition of $\alpha_\TSP$ (resp. $\alpha_\CVRP$), we do not know if an $\alpha_\TSP$- (resp. $\alpha_\CVRP$-) approximation for TSP (resp. CVRP) exists or not. So, we need to include an $\epsilon$ term in the definition of $\alpha$ in the the second and third cases.}

\begin{algorithm}[ht]
	\caption{Online FDP on General Metrics without or with Speeding}
	\label{alg:multi-general}
		\begin{algorithmic}[1]
			\State let $\gamma \gets \frac{2\alpha + 2}{\epsilon}$, where
			\begin{align*}
				\alpha \gets 
				\begin{cases}
					1 & \text{if running time of online algorithm can be exponential}\\
					\alpha_\TSP + \epsilon & \text{if running time must be polynomial and } c = \infty\\
					\alpha_\CVRP + \epsilon & \text{if running time must be polynomial and } c \neq \infty
				\end{cases}
			\end{align*}
			\For{$i = 1,2,3,4,\cdots$}
				\State wait until time $i \gamma F$, let $R' \gets \big\{\rho \in R: r_\rho \in [(i-1)\gamma F, i\gamma F)\big\}$
			    \State find an $\alpha$-approximate CVRP tour $(o,u_1,u_2,\cdots,u_z,o)$  for $R'$ with capacity $c$, such that the set $Q_{z'} \subseteq R'$ of requests is served when the vehicle visits $u_{z'}$, for every $z' \in [z]$. \label{step:speeding-find-tour}
          \State mark all vehicles as free, $j'\gets 1$
          \While{$j' \leq z$}
            \State $j \gets $ the largest integer in $[j', z]$ satisfying that $(u_{j'}, u_{j'+1}, \cdots, u_j)$ has length at most $\alpha(\gamma + 2)F$
              \State use one free vehicle to take the tour $(o,u_{j'},u_{j'+1},\cdots,u_{j},o)$ at time $i \gamma F$,  and serve $Q_{z'}$ when it visits $u_{z'}$ for every $z' \in [j', j]$ \label{step:speeding-use-free}
              \State  mark the vehicle as busy, $j' \gets j+1$
          \EndWhile
			\EndFor
		\end{algorithmic}
	\end{algorithm}

Notice that in the tour $(o, u_1, u_2, \cdots, u_z, o)$, we are guaranteed that the number of requests served between two adjacent visits of $o$ is at most $c$.  Therefore, this condition also holds for every tour $(o, u_{j'}, u_{j'+1}, \cdots, u_j, o)$ constructed in step~\ref{step:speeding-use-free}, meaning that every tour constructed in the step is valid. 

\begin{lemma}
  \label{lem:general-multi-ratio}
  In step~\ref{step:speeding-use-free}, we always have a free vehicle available. Moreover, with $(\alpha + \epsilon)$-speeding, the vehicle can complete the tour at or before time $(i+1)\gamma F$.
\end{lemma}

\begin{proof}
 
  We claim that we always have a free vehicle in step~\ref{step:speeding-use-free}. Otherwise, there are integers $1 = j_0 < j_1 < j_2 < \cdots < j_k < z$ such that for every $o \in [k]$ the length of $(u_{j_{o-1}},u_{j_{o-1}+1}, u_{j_{o-1}+2}, \cdots, u_{j_o})$ is more than $\alpha(\gamma+2)F$. So, the tour $(o, u_1, u_2, \cdots, u_z, o)$ has total length more than $\alpha k(\gamma + 2)F$, contradicting Lemma~\ref{lem:bound-tsp} and that the tour is $\alpha$-approximate.  
  
  Then we claim each tour $(u_{j'},\cdots,u_{j})$ assigned to a vehicle in step~\ref{step:speeding-use-free} has length at most $\alpha(\gamma + 2)F + 2F$. This holds as the length of $(u_{j'},\cdots,u_{j})$ is at most $\alpha(\gamma + 2)F$. Moreover $d(o, u_{j'}) \leq F$ and $d(u_{j}, o) \leq F$. Now with $\alpha+\epsilon$ speed, the travel time of each vehicle is at most 
  $
  \frac{\alpha(\gamma + 2)F + 2F}{\alpha+\epsilon}  = \gamma F.
  $
  Therefore, we proved that all vehicles can come back at or before $(i+1)\gamma F$. 
\end{proof}

It implies that all the requests released in $[(i-1)\gamma F, i \gamma F)$ can be served at or before $(i+1)\gamma F$, which concludes the $2\gamma = \frac{4\alpha + 4}{\epsilon}$ competitive ratio.  For (\ref{thm:ALG-speeding}a), we allow the algorithm to run in exponential time, so the $1$-approximate CVRP tour in step~\ref{step:speeding-find-tour} can be found. If the algorithm is only allowed to run in polynomial time, then the $\alpha$-approximate CVPR tour can be found, by the definition of $\alpha_\TSP$ and $\alpha_\CVRP$. The final speeding factor is $\alpha_\TSP + 2\epsilon$ or $\alpha_\CVRP + 2\epsilon$ depending on whether $c = \infty$.  This implies (\ref{thm:ALG-speeding}b) and (\ref{thm:ALG-speeding}c) by replacing $\epsilon$ in our proof with $\epsilon/2$.

    \bibliography{main}
    \bibliographystyle{alpha}
    
    \appendix
    \section{Guessing an Upper Bound $F$ on Optimum Maximum Flow Time Online}
\label{appendix:guessing}

In this section, we show how to use the $\beta$-competitive algorithm $\calA$ for the case when $F$ is given, to design an $8\beta$-competitive algorithm $\calA'$ for the case when $F$ is not given, with the same speeding factor and running time.  Initially let $F$ be the smallest possible non-zero flow time\footnote{Notice that the algorithm can easily check if the optimum maximum flow time is $0$ or not.} and $D \gets 0$. $D$ will denote be the delay of all requests: In the algorithm, we let the delayed arrival time of a request $\rho$ to be $r_\rho + D$. The algorithm runs in stages. In each stage, we follow algorithm $\calA$ on the current $F$, using the delayed arrival times for requests. Suppose at some time point $t$, we found that the flow time of some request is about to exceed $\beta F$, or we are starting a trip of length more than $2 \beta F$. Then we know that our guessed $F$ does not upper bound the optimum maximum flow time. In this case, we then we wait for all the current active trips to complete, update $F \gets 2F$ and $D \gets D + 3\beta F$, and start a new stage. Since we increased $D$ by $3\beta F$, no requests have arrived before the beginning of the new stage, per the delayed arrival times. 

Assume $F^*$ is the optimum maximum flow time for the instance. Then our final $F$ will be at most $2F^*$. With the delayed arrival times, all requests will have flow time at most $2\beta F^*$. The final $D$ is at most $3\beta \left(F^* + \frac{F^*}{2} + \frac{F^*}{4} + \cdots\right) = 6\beta F^*$. So, all requests will have flow time at most $2\beta F^* + 6\beta F^* = 8\beta F^*$.

    \section{Hardness of Offline Food Delivery Problem}
\label{sec:hardness}
In this section, we show hardness of approximating the offline food delivery problem, without and with speeding, finishing the proof of (\ref{thm:LB}a), (\ref{thm:LB-speeding}b) and (\ref{thm:LB-speeding}c).

\subsection{$\Omega(n)$-Hardness for offline FDP without Speeding: Proof of (\ref{thm:LB}a)}
\label{subsec:LB-offline}

In this section, we prove (\ref{thm:LB}a) by reducing the Hamiltonian cycle problem to offline FDP.  Let ${\bar G} = ({\bar V}, {\bar E})$ be a Hamiltonian cycle instance, with $|{\bar V}| = \bar n$ and $o \in {\bar V}$. 
In our FDP instance, the graph $G = (V, E)$ is obtained from ${\bar G} = (\bar V, \bar E)$ by adding ${\bar n}$ new vertices $u_1,u_2, \cdots, u_{\bar n}$ and ${\bar n}$ edges $(o, u_1), (o, u_2), \cdots, (o, u_{\bar n})$. We associate a length function $\ell$ over $E$: the edges in $\bar E$ have length $1$, and the $\bar n$ new edges have length $\bar n$.  The size of the new graph is $n := |V| = \bar n + \bar n = 2\bar n$.

There are ${\bar n}^2$ phases of the timeline, each with length $(2{\bar n}+1){\bar n}$. So, phase $h$ starts at time $(h-1)(2{\bar n}+1){\bar n}$ and ends at $h(2{\bar n}+1){\bar n}$.  The requests arrived at phase $h$ are as follows, where for simplicity we let $b_h = (h-1)(2{\bar n}+1){\bar n}$ denote the starting time of phase $h$: 
\begin{align*}
    \left\{\big(b_h, v\big): v \in \bar V \setminus \{o\}\right\} \cup \left\{ \big(b_h + (2i-1)\bar n, u_i\big): i \in [\bar n] \right\}.
\end{align*}
So, at the beginning of the phase, requests arrive from $\bar V \setminus \{o\}$. Then starting in time $b_h + \bar n$, requests come from the $\bar n$ new vertices, one after another in every $2\bar n$ units time. 
This finishes the description of the offline FDP instance. 

If $\bar G$ contains a Hamiltonian cycle, then the correspondent FDP instance has a solution with maximum flow time at most ${\bar n}$: in each phase $h \in [{\bar n}^2]$, the solution uses the Hamiltonian cycle to serve the $({\bar n}-1)$ requests for ${\bar V} \setminus \{o\}$ from time $b_h$ to time $b_h + \bar n$.  Then for each $i \in [{\bar n}]$, from time $b_h + (2i-1){\bar n}$ to $b_h + (2i+1){\bar n}$, it takes the trip $(o, u_i, o)$. Clearly, the maximum flow time of the solution is ${\bar n} = O(n)$. 

Now suppose $\bar G$ does not contain a Hamiltonian cycle. We show no trip in a solution serves two requests in two different phases, unless it has maximum flow time at least $(2{\bar n}+1){\bar n} = \Omega(n^2)$.  Notice that the removal of $o$ from $G$ results in ${\bar n}+1$ induced sub-graphs: ${\bar G} \setminus o$ and the ${\bar n}$ singletons $u_1, u_2, \cdots, u_{\bar n}$. Then by our definition of trips, no trip can serve requests from two different sub-graphs. However, for two requests from the same sub-graph but arrive in different phases, their arrival times are at least one phase-time (i.e, $(2{\bar n}+1){\bar n}$ units time) apart.  Therefore if a trip serves the two requests, the maximum flow time of the solution is at least $(2{\bar n}+1){\bar n}$.

Therefore, we can assume in a solution, we handle each phase separately. It takes at least $\bar n + 1 + \bar n \cdot 2\bar n = (2{\bar n}+1){\bar n} + 1$ units time to serve all the requests in a phase. For all the ${\bar n}^2$ phases, the trips take a total time of at least ${\bar n}^2((2{\bar n}+1){\bar n} + 1)$. As the last request arrives at time ${\bar n}^2((2{\bar n}+1){\bar n}) - {\bar n}$, the maximum flow time of the solution is at least 
${\bar n}^2((2{\bar n}+1){\bar n} + 1) - {\bar n}^2((2{\bar n}+1){\bar n}) = {\bar n}^2 = \Omega(n^2)$.  This finishes the proof of the first statement Theorem~\ref{thm:LB}.

\subsection{$\Omega(\sqrt{n})$-Hardness With Speeding: Proof of (\ref{thm:LB-speeding}b) and (\ref{thm:LB-speeding}c)}
\label{subsec:LB-offline-speeding}
Now we move to the hardness of the offline food delivery problem with speeding. We first consider the version where $c = \infty$, proving (\ref{thm:LB-speeding}b). One natural modification we need is to use the TSP instance for the proof of hardness of approximation, to replace that for NP-hardness (i.e, the Hamiltonian-Cycle instance). However, the $\bar n$ legs in $G$ will create an issue in the proof as with speeding, we need much less time to travel through the legs, and the time saved is sufficient to offset the time difference between yes and no instances.   Instead, our graph $G$ is obtained by taking ${\bar n}$ copies of ${\bar G}$ and identify the ${\bar n}$ copies of $o$. As the resulting graph $G$ has $n = \Theta({\bar n}^2)$ vertices, we only obtain a $O(\sqrt{n})$-hardness. 

 Let $\alpha'$ be an absolute constant in $(\alpha, \alpha_{\TSP})$. Consider any TSP instance ${\bar G}=({\bar V}, {\bar E})$ with ${\bar n} = |{\bar V}|$ and positive edge lengths, and assume $o \in {\bar V}$.  Let $L$ be a lower bound on the length of the optimum TSP tour of $\bar G$. 
 We need to show that for some $f(n) = \Omega(\sqrt{n})$ the following statement holds. If there is a polynomial time $\alpha$-speeding $f(n)$-approximation algorithm for offline FDP,  then we can find a TSP tour of $\bar G$ of length at most $\alpha' L$ in polynomial time. Then this will give us a $(\alpha' + \epsilon)$-approximation for TSP for any constant $\epsilon > 0$, which contradicts the definition of $\alpha_\TSP$.

Now we start to define the instance. We make ${\bar n}$ copies of ${\bar G}$ and identify all ${\bar n}$ copies of $o$. So, the resulting graph has $n:={\bar n}({\bar n}-1)+1 = \Theta(n'^2)$ vertices.  For every $v \in \bar V \setminus \{o\}$ and $i \in [n']$, we use $v^{(i)}$ to denote the copy of $v$ in the $i$-th copy of $\bar G$. There are ${\bar n}$ phases of the timeline, each with length $L$. The requests that arrive in phase $h \in [\bar n]$ are as follows:
\begin{align*}
    \left\{\big((h-1)L, v^{(h)}\big): v \in \bar V \setminus \{o\} \right\}.
\end{align*}
That is, at the beginning of $h$-th phase, the depot receives requests from all non-depot vertices in the $h$-th copy of $\bar G$. 

Since there is a TSP tour of length at most $L$ for $\bar G$, the FDP instance has a solution with maximum flow time at most $L$: In each phase $h \in [\bar n]$, the solution uses the optimum TSP tour to serve all the requests arrived at the beginning of the phase and return to the depot $o$. 

Suppose our FDP algorithm finds an $\alpha$-speed $f(n)$-approximate solution in polynomial time.  Then the total time spent for all trips by the solution is at most $\bar n L + f(n) L$, with the $\alpha$-speeding.  The total length of all trips is at most $\alpha(\bar n L + f(n) L)$ (we use length for travel time without speeding). Moreover, the trips form a TSP tour of $G$, and thus contains $\bar n$ disjoint TSP tours of $\bar G$. The shortest TSP tour has length at most $\frac{\alpha(\bar n L + f(n) L)}{\bar n} = \alpha L + \frac{f(n)}{\bar n}L$.  If we have $f(n) \leq (\alpha' - \alpha) \bar n$, then the length is at most $\alpha' L$. As $\alpha' - \alpha$ is a constant and $\bar n = \Theta (\sqrt{n})$,  we can require $f(n) \leq (\alpha' - \alpha) \bar n$ while still guaranteeing $f(n) = \Omega(\sqrt{n})$. This finishes the proof of (\ref{thm:LB-speeding}b).
\medskip

The proof immediately extends to (\ref{thm:LB-speeding}c).  We only need to change $\alpha_\TSP$ to $\alpha_\CVRP$, and replace the TSP instance $\bar G = (\bar V, \bar E)$ with a CVRP instance $(\bar G = (\bar V, \bar E), c)$.

\section{Lower Bounds on Competitive Ratio of Online Algorithms for Single-Vehicle Uncapacitated FDP}
\label{sec:LB}
In this section, we show lower bounds on the competitive ratios of online algorithms for single-vehicle uncapacitated FDP without or with $(1+\epsilon)$-speeding, proving (\ref{thm:LB}b) and (\ref{thm:LB-speeding}a).  We first give a base instance in Section~\ref{subsec:base}, and the final instances for the lower bounds will be obtained by repeating the base instance many times in different ways.  Then in Sections~\ref{subsec:LB-online} and \ref{subsec:LB-online-speeding} we prove the two statements respectively. 

\subsection{Base Instance}
\label{subsec:base}
In this section, we describe a base FDP instance that will be used in the proofs of (\ref{thm:LB}b) and (\ref{thm:LB-speeding}a).   Let $p \geq 2$ be an integer. The graph ${\bar G} = ({\bar V}, {\bar E})$ with the depot $o \in {\bar V}$ is defined as follows. We have ${\bar V} = \{o\} \cup \{v^{is}_j: i \in [0, p], j \in [7], s \in \{\sfL, \sfR\}\}$, where for every $j \in [7]\setminus\{4\}$, $v^{0\sfL}_j$ and $v^{0\sfR}_j$ are identified.  So we have $\bar n := |{\bar V}| = 1 + 2\times 7(p+1) - 6 = 14p + 9$. The set ${\bar E}$ contains two types of edges: the depot edges, which are incident to the depot $o$ and have length $1/2$, and the path edges, which are not incident to $o$ and have length $1$.  For every $i \in [0, p], j \in \{1, 7\}$ and $s \in \{\sfL, \sfR\}$, there is a depot edge $(o,v^{is}_j)$. So there are $4(p+1)-2 = 4p + 2$ depot edges.   The path edges are as follows:
\begin{itemize}[itemsep=0pt,parsep=0pt]
	\item For every $i \in [0, p]$ and $s \in \{\sfL, \sfR\}$, and $j \in [6]$, there is a path edge $(v^{is}_j, v^{is}_{j+1})$.
	\item For every $i \in [p]$ and $s \in \{\sfL, \sfR\}$, there are two path edges $(v^{is}_3, v^{(i-1)s}_4)$ and $(v^{(i-1)s}_4, v^{is}_5)$.
\end{itemize}
This completes the description of the graph ${\bar G} = ({\bar V}, {\bar E})$ with edge lengths. See Figure \ref{fig:grid} for a depiction of the graph.

\begin{figure}[t]
\centering
\includegraphics[width=\textwidth]{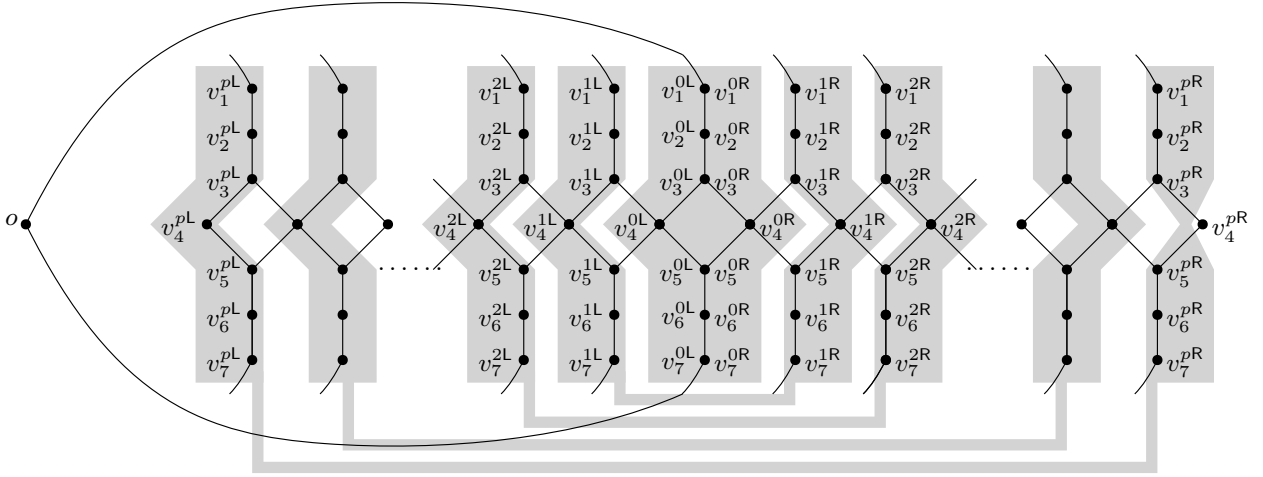}
\caption{The graph ${\bar G} = ({\bar V}, {\bar E})$ used for the lower bounds on the competitive ratio of online algorithms without speeding and with $(1+\epsilon)$-speeding. The depot edges (only two of which are shown) have length $1/2$ and the path edges have length $1$. $v^{0\sfL}_j = v^{0\sfR}_j$ for every $j \in [7]\setminus\{4\}$. Each gray polygon shows a set of positions that are requested in a same segment of a left-leaned base instance. } 
\label{fig:grid}
\end{figure}

It will be convenient for us to define two sets of paths as follows:
\begin{itemize}[itemsep=0pt,parsep=0pt]
	\item $P^{is}$ for every $i \in [0, p]$ and $s \in \{\sfL, \sfR\}$,  which is the path between $v^{is}_1$ and $v^{is}_7$, containing edges $(v^{is}_j, v^{is}_{j+1})$ for every $j \in [6]$, and
	\item $Q^{is}$ for every $i \in [p]$ and $s \in \{\sfL, \sfR\}$,  which is the path $(v^{is}_1,  v^{is}_2,  v^{is}_3, v^{(i-1)s}_4,  v^{is}_5,  v^{is}_6,  v^{is}_7)$.
\end{itemize}\medskip

Now we define the set of requests in the base instance.  The length of the whole time horizon of the base instance is $7(2p + 1)$, which is divided into $p+1$ segments indexed by $0, 1, 2, \cdots, p$. Segment $0$ has length $7$ and each segment $i \in [p]$ has length $14$. For every $i \in [0, p]$, let $b_i$ denote the beginning time of segment $i$. Thus $b_0 = 0$ and $b_i = 7(2h-1)$ for every $i \in [p]$. The following set of all requests are always included in the base instance:
\begin{align*}
    \bar R :=\Big\{\big(b_i, v^{hs}_{j}\big): i \in [p], s \in \{\sfL, \sfR\}, j \in [7]\Big\} \quad \bigcup \quad \Big\{\big(b_p, v^{ps}_j\big): s\in \{\sfL, \sfR\}, j \in [7] \setminus \{4\}\Big\}.
\end{align*}
For $s \in \{\sfL, \sfR\}$, we define 
\begin{align*}
    \bar R_s := \bar R \cup \Big\{\big(b_p, v^{ps}_4\big)\Big\}.
\end{align*}

The set of requests in the base instance is either $\bar R_\sfL$, in which case we say the instance is left-leaned, or $\bar R_\sfR$, in which case we say it is right-leaned.

Therefore, at the beginning of segment $i$, requests come from vertices of the form $v^{is}_j$. The only exception is segment $p$, for which there is no request from $v^{p\sfR}_4$ or $v^{p\sfL}_4$, depending on whether the base instance is left or right-leaned. Notice that the numbers of requests that arrive at the beginning of segment $0$, segment $i \in [p-1]$ and segment $p$ are respectively $8, 14$ and $13$. 
This finishes the description of the base instance.  See Figure \ref{fig:grid} again for the set of requests that arrive in each segment.

\begin{lemma}
	There is a solution to the base instance with $O(1)$ maximum flow time, where all trips complete by time $7(2p+1)$. 
\end{lemma}
\begin{proof}
	Assume the base instance is left-leaned. The solution will follow the trip $(o, P^{0\sfL}, o)$ in segment $0$ to serve all requests that arrived at time $0$ except $(0, v^{0\sfR}_4)$. Then in each segment $i \in [p]$, it takes the two trips $(o, Q^{i\sfR}, o)$ and $(o, P^{i\sfL}, o)$, to serve $(b_{i-1}, v^{(i-1)\sfR}_4)$, and all the requests that arrived at time $b_i$ except $(b_i, v^{i\sfR}_4)$. Notice that at the end of segment $p$, no request will be left unserved, as there is no request from $v^{p\sfR}_4$.  Then every request has flow time at most 2 segment lengths, which is $O(1)$. All trips complete by time $7(2p+1)$. The solution for a right-leaned base instance can be defined symmetrically. 
\end{proof}

Now consider the online algorithm, which does not know if the base instance is left or right-leaned.  Then it can not make the correct decision unless it waits for almost the whole time horizon, in which case a large flow time will be incurred.  We show that if the algorithm made a wrong decision, then the total length of trips it uses to serve all the requests is at least $7(2p+1) + 1$.  This will create a delay of $1$. By repeating the base instance many times, the delay will accumulate. 

\begin{lemma}
	\label{lemma:follow-paths-L}
	Let ${\bar V}' = {\bar V} \setminus \{v^{p\sfR}_4\}$. Then the shortest TSP tour for ${\bar V}'$ has length $7(2p+1)$, and its edge multi-set is ${\bar E}' := (\text{depot edges}) \uplus \left(\biguplus_{i \in [0, p]} P^{i\sfL}\right) \uplus \left(\biguplus_{i \in [p]}Q^{i\sfR}\right)$. 
\end{lemma}
Notice that we only allow a TSP tour to use edges in the graph ${\bar G}$. Thus, in general its edge set can be a multi-set. However the lemma says the optimum TSP tour has the edge set being an ordinary one, and $\bar E'$ is the unique optimum TSP tour, if we ignore the order of traversing the edges.
\begin{proof}[Proof of Lemma \ref{lemma:follow-paths-L}]
	First, $|{\bar V}' \setminus \{o\}| = 7(2p+1)$ and the distance between any two distinct non-depot vertices is at least $1$. So the TSP for ${\bar V}'$ is at least $7(2p+1)$.  Clearly, ${\bar E}'$ gives a TSP tour for ${\bar V}'$ of  length $7(2p+1)$. It remains to show that ${\bar E}'$ is the unique optimum TSP tour.  Notice that an optimum TSP tour for ${\bar V}'$ can not visit a vertex in ${\bar V}' \setminus \{o\}$ twice, and can not visit the vertex $v^{p\sfR}_4$ which is not in ${\bar V}'$; it can visit $o$ multiple times.  
	
	Notice that in the graph ${\bar G}[{\bar V}']$ (sub-graph of ${\bar G}$ induced by ${\bar V}'$),  $v^{p\sfL}_2, v^{p\sfL}_4$ and  $v^{p\sfL}_6$ have degree $2$ and there are incident to non-depot vertices. That means their incident edges must be chosen in the tour.  In other words, we have to choose the path $P^{p\sfL}$ in the tour. Then, we need to choose the two depot edges incident to $v^{p\sfL}_1$ and $v^{p\sfL}_7$. As we completed a cycle $(o, P^{p\sfL}, o)$, we can remove vertices in $P^{p\sfL}$ from ${\bar V}'$. Then, we can argue that the tour will contain the cycle $(o, P^{(p-1)\sfL}, o)$ and we can remove vertices in $P^{(p-1)\sfL}$ from ${\bar V}'$. Repeating this argument, the tour will contain the cycles $(o, P^{(p-2)\sfL}, o), (o, P^{(p-3)\sfL}, o), \cdots, (o, P^{0\sfL}, o)$.  Continuing the process on the right-side of the figure, we can show that the tour must contain the cycles $(o, Q^{1\sfR}, o), (o, Q^{2\sfR}, o), (o, Q^{3\sfR}, o), \cdots, (o, Q^{p\sfR}, o)$.
\end{proof}

The following lemma can be proved in a symmetric way:
\begin{lemma}
	\label{lemma:follow-paths-R}
	Let ${\bar V}' = {\bar V} \setminus \{v^{p\sfL}_4\}$. Then the shortest TSP tour for ${\bar V}'$ has length $7(2p+1)$, and its edge multi-set is ${\bar E}' := (\text{depot edges}) \uplus \left(\biguplus_{i \in [0, p]} P^{i\sfR}\right) \uplus \left(\biguplus_{i \in [p]}Q^{i\sfL}\right)$. 
\end{lemma}

\subsection{$\Omega(n)$-Lower bound on Competitive Ratio of Online Algorithms without Speeding: Proof of (\ref{thm:LB}b)}
\label{subsec:LB-online}

With the base instance defined, we show how to repeat it to obtain our lower bounds for competitive ratios.  First focus on the no-speeding setting.  In this case, the graph $G = (V, E)$ will be obtained from $\bar G$ by adding $p$ legs of length $1/2$: $V = \bar V \cup \{u_1, u_2, \cdots, u_p\}, E = \bar E \cup \{(o, u_1), (o, u_2), \cdots, (o, u_p)\}$, where $u_i$'s are new vertices we create. The length of each $(o, u_i)$ is $1/2$. Notice that $n:=|V|=15p+9$.

For convenience, we let $L = 7(2p+1)$ be the length of the time horizon of the base instance.  Now we define the requests in our instance for the $\Omega(n)$-lower bound.  The time horizon is divided into $p$ phases of length $L + p$ each.  Each phase $h \in [p]$ may be left-leaned, in which case we let $s_h = \sfL$, or right-leaned, in which case we let $s_h = \sfR$.  The requests that arrive in phase $h$ is as follows:
\begin{align*}
    \Big\{ \big((h-1)(L+p) + r_\rho, v_\rho\big): \rho \in \bar R_{s_h}\Big\} \quad \bigcup \quad \Big\{\big((h-1)(L+p) + L + j - 1, u_j \big):j \in [p]\Big\}.
\end{align*}

So, in phase $h$, we ``copy'' the requests in $\bar R_{s_h}$ to the first $L$ time units of the phase. Then starting at time $L$ of the phase, requests from $\{u_1, u_2, \cdots, u_p\}$, i.e, the end points of the legs, arrive one by one, in gap of 1 time unit.  This finishes the definition of the instance $\calI$.  Notice that there are $2^p$ possibilities for $\calI$, decided by the vector $(s_h)_{h \in [p]}$.

Clearly, there is a solution with maximum flow at most $O(1)$ for the instance $\calI$, for every $(s_h)_{h \in [p]}$. In each phase $h$, we just follow the solution for the left-leaned or right-leaned base instance, depending on whether $s_h \in \{\sfL, \sfR\}$. The trips take a total time of $L$. Then we visit $u_1, u_2, \cdots, u_p$ one by one in that order. 

Now we can prove the lemma lower bound the online algorithm.
\begin{lemma}
	Any deterministic online FDP algorithm must incur a maximum flow time of $p - O(1)$, for the instance $\calI$ for some $(s_h)_{h \in [p]}$.
\end{lemma}

\begin{proof}
    Consider two requests in different phases with delivery locations in $\bar V \setminus \{o\}$. Their arrival times must differ by at least $p$.  For any $j \in [p]$ and two requests with delivery location $u_j$ in different phases, their arrival times must differ by at least $L+p$.  Then using the similar argument as in Section~\ref{subsec:LB-offline}, the online algorithm will not using a single trip to serve two requests from two different phases, unless it has a flow time of at least $p$.
	
	Now focus on a phase $h \in [p]$. We can assume the online algorithm served the request $\big((h-1)(L+p), v^{0\sfL}_4\big)$ by the time $(h-1)(L+p) + p$, since otherwise the flow time of the request is already $p$. This time is before the time $(h-1)(L+p) + L - 14$ when the adversary reveals whether the phase is left-leaned or right-leaned.  If the online algorithm uses $(v^{0\sfL}_3, v^{0\sfL}_4)$ or $(v^{0\sfL}_4, v^{0\sfL}_5)$ to enter $v^{0\sfL}_4$ to satisfy the request, then the adversary will make the phase right-leaned. Otherwise, the online algorithm uses $(v^{1\sfL}_3, v^{0\sfL}_4)$ or $(v^{0\sfL}_4, v^{1\sfL}_5)$ to enter $v^{0\sfL}_4$. In this case the adversary will make the phase left-leaned. Then in either case, the trips used to satisfy all requests in phase $h$ is at least $L+1 + p$, by Lemma~\ref{lemma:follow-paths-L} and \ref{lemma:follow-paths-R}. Thus, overall, the total length of all trips over all phases is at least $p(L+1+p)$, meaning that the last request served by the online algorithm must have flow time at least $p(L+1+p) - p(L+p) - O(1) = p - O(1)$. 
\end{proof}

This finishes the proof of (\ref{thm:LB}b) as we can set $p = \Theta(n)$. 

\subsection{$\Omega(\frac1\epsilon)$-Lower bound on Competitive Ratio of Online Algorithms with $(1+\epsilon)$ Speeding: Proof of (\ref{thm:LB-speeding}a)}
\label{subsec:LB-online-speeding}

In the instance, we let $p  = \frac{1}{20\epsilon} + 0.05$ so that $1+\epsilon = \frac{p}{p-0.05}$. Assume $\epsilon$ is small enough so $p$ is a large enough integer.   The graph $G = (V, E)$ is obtained by making $p$ copies of the base graph $\bar G$ and identifying the $p$ copies of $o$.  So, we have $n:=|V| = p(\bar n - 1) + 1 = p(14p + 8) - 1$. For every $v \in \bar V \setminus \{o\}$ and $i \in [p]$, we let $v^{(i)}$ denote the copy of $v$ in the $i$-th copy of $\bar G$. 

The requests of our instance is also the concatenation of $p$ copies of the requests from the base instance. The time horizon has length $pL$ and is divided into $p$ phases of length $L$ each, where $L=7(2p+1)$ is the length of the time horizon for the base instance. The adversary can choose $s_h \in \{\sfL, \sfR\}$ for every phase $h \in [p]$, i.e., choose whether phase $h$ is left or right-leaned. The requests for phase $h$ is defined as
\begin{align*}
    \Big\{\big((h-1)L + r_\rho, v_\rho^{(h)}\big): \rho \in \bar R_{s_h}\Big\}.
\end{align*}
That is, we ``copy'' the requests in the left or right-leaned base instance to phase $h$, depending on whether the phase is left or right-leaned.

Since the $p$ copies of $\bar G$ are connected only through $o$, no trips will serve two requests from two different phases.   The online algorithm must serve the request from $(v^{0\sfL}_4)^{(h)}$ a phase $h \in [p]$, before it knows if the phase is left-leaned or right-leaned. Otherwise the flow time of the request is at least $L -O(1) = \Omega(p) = \Omega(1/\epsilon)$; notice that this argument is irrespective of the speeding factor we use. The adversary can follow the same strategies as in Section~\ref{subsec:LB-online}. With $(1+\epsilon)$-speeding, the total time needed for the online algorithm to complete all trips is at least $\frac{p(L+1)}{1+\epsilon} = \big(p(L+1)\big)\cdot \frac{p-0.05}{p} = (p-0.05)(L+1) = p\cdot 7(2p+1)- 0.35(2p+1) + p-0.05  = 7p(2p+1) + 0.3 p - 0.4$. Therefore, some request served the last must have flow time at least $0.3p - O(1) = \Omega(1/\epsilon)$.  This finishes the proof of (\ref{thm:LB-speeding}a).
    \section{$\Omega\left(\sqrt{|R|}\right)$ Lower Bound on Competitive Ratios of Online (Capacitated) FDP on Tree Metrics: Proof of Theorem~\ref{thm:LB-capacity}}
\label{sec:LB-capacity}


In this section, we prove Theorem~\ref{thm:LB-capacity}. 
 The idea is similar to the construction for the uncapacitated case on general graphs. In each phase, the online algorithm is forced to choose between two options without knowing which one is the right one. If it makes the wrong choice, the total time spent on trips will be increased by at least 1. Then the delay of 1 in each phase will accumulate, resulting in a large flow time. 
 
 The tree is simple: we have $5$ vertices $\{o, v_1, v_2, v_3, v_4\}$, and $4$ edges $(o, v_1), (v_1, v_2), (v_1, v_3)$, and $(o, v_4)$, each with length 1. There is only $k=1$ vehicle, and the capacity of the vehicle is $c=2$. 
 
 Let $p \geq 2$ be an integer parameter.
%
 %
%
There are $p$ phases in the timeline, each with length $10p$. So, phase $h \in [p]$ starts at time $b_h: = 10(h-1)p$ and ends at $10hp$. In the phase, the adversary can choose a $v^* \in \{v_2, v_3\}$, and release the following set of requests (notice that with capacities, the set of requests may be a multi-set): 
\begin{align*}
    & \Big\{\big(b_h, v_1\big), \big(b_h, v_2\big), \big(b_h, v_3\big)\Big\} \quad \bigcup \quad \Big\{\big(b_h + 8j - 4, v\big):j \in [p-1], v \in \{v_2, v_2, v_3, v_3\}\Big\}\\
    &\qquad  \bigcup \quad \Big\{\big(b_h+8p-4, v^*\big)\Big\}
    \quad \bigcup \quad \Big\{\big(b_h + 8p + 2(j-1), v_4\big), \big(b_h + 8p + 2(j-1), v_4\big): j \in [p]\Big\}.
\end{align*}

So, at the beginning of the phase, the adversary releases 3 requests, one from each of $\{v_1, v_2, v_3\}$. Then starting from time $b_h + 4$, in every $8$ time units, the adversary releases 4 requests, 2 at $v_2$ and 2 at $v_3$. This is repeated for $p-1$ times; so the last set of 4 requests arrive at time $b_h + 8p-12$. Then at time $b_h + 8p -4$, it releases a request at $v^*$, which is either $v_2$ or $v_3$. Starting from time $b_h + 8p$ in every 2 time units, the adversary releases two requests from $v_4$, until time $b_h + 10p - 2 = b_{h+1}-2$.

Now the optimum solution has flow time $O(1)$. Focus on a phase $h \in [p]$ and suppose $v^* = v_2$.  Then in the interval $[b_h, b_h+4]$, the online algorithm will use a trip to serve the two requests $(b_h, v_1)$ and $(b_h, v_3)$. Then the request $(b_h, v_2)$ is not served yet. Then for each $j \in [p-1]$, from $b_h + 4 + 8(j - 1)$ to $b_h + 4 + 8j$, the online algorithm will use 2 trips of length 4 each to satisfy 4 requests: $\big(b_h + \max\{8j - 12, 0\}, v_2\big)$, $\big(b_h + 8j - 4, v_2\big)$, and two requests $(b_h+8j-4, v_3)$. From time $b_h + 4 + 8(p-1) = b_h + 8p-4$ to $b_h + 8p$, it uses a trip of length $4$ to serve $\big(b_h+8p-12, v_2\big)$ and $\big(b_h+8p-4, v_2\big)$. Finally, from time $b_h + 8p + 2(j-1)$ to $b_h + 8p + 2j$ for every $j \in [p]$, the algorithm uses a trip of length 2 to serve the two requests $(b_h + 8p +2(j-1), v_4)$. Notice that all requests are served with a flow time of at most $16$. The solution for a phase $h$ with $v^* = v_3$ can be defined symmetrically. 

We now argue that an online algorithm must have a maximum flow time of $p$. Now focus on any phase $h \in [p]$.  Similar to the arguments we used before, we can argue that the online algorithm will not use a trip to serve two requests arrived in two different phases, unless it has a flow time of at least $2p$.

Fix a phase $h \in [p]$, and assume for the phase we have $v^* = v_2$ and focus on the requests from $v_1, v_2$ and $v_3$.  There are 1 request from $v_1$, $2p$ requests from $v_2$, and $2p-1$ requests from $v_3$. Then with a capacity-$2$ vehicle, it takes a minimum of $8p$ units time to serve these requests.  Through a simple case-by-case analysis, one can show that to get a time of $8p$, the matching of the requests must be the following: match the request from $v_1$ with a request from $v_3$, match $p$ pairs of requests from $v_2$, and $p-1$ pairs of requests from $v_3$.   Similarly, if $v^* = v_3$, then to obtain a time of $8p$ to serve the requests, we need to match the request for $v_1$ with a request from $v_2$, match $p-1$ pairs of requests from $v_2$, and $p$ pairs of requests from $v_3$. 

We can now design the strategy for the adversary. In each phase $h \in [p]$ it observes how the online algorithm serves the request $(b_h, v_1)$. If the request is served at or after time $b_h + 8p - 4$, then the flow time of the request is already more than $p$. Assume otherwise.  If the online algorithm serves the request together with some request from $v_2$, then the adversary will let $v^* = v_2$. Otherwise, it will let $v^* = v_3$. Then the total length of all trips serving requests from $\{v_1, v_2, v_3\}$ in the phase will be at least $8p+2$. The total length of all trips for the requests in the phase is at least $10p+2$.  So, over all the $p$ phases, the total length is at least $p(10p+2) = 10p^2 + 2p$. Then the flow time of the last request served is at least $10p^2 + 2p -10p^2 - O(1) = 2p-O(1) \geq p$.  This finishes the proof of Theorem~\ref{thm:LB-capacity} since in our instance we have $p = \Theta(\sqrt{|R|})$.

\end{document}